\documentclass[twoside,11pt]{article}

\usepackage[preprint]{jmlr2e}
\usepackage[utf8]{inputenc}
\usepackage{xcolor}
\usepackage{svg}
\usepackage{amsmath,bm}
\usepackage{amsfonts}
\usepackage[margin=2.54cm]{geometry}
\usepackage{appendix}
\usepackage{footnotehyper}
\usepackage{enumitem}
\usepackage[TS1,T1]{fontenc}
\usepackage{textcomp}
\usepackage{caption}
\usepackage{subcaption}
\usepackage{booktabs}
\usepackage{tikz}
\usepackage{etoolbox}
\usetikzlibrary{calc,fadings,decorations.pathreplacing}

\newcommand\pgfmathsinandcos[3]{%
  \pgfmathsetmacro#1{sin(#3)}%
  \pgfmathsetmacro#2{cos(#3)}%
} 
\newcommand\LongitudePlane[3][current plane]{%
  \pgfmathsinandcos\sinEl\cosEl{#2} 
  \pgfmathsinandcos\sint\cost{#3} 
  \tikzset{#1/.estyle={cm={\cost,\sint*\sinEl,0,\cosEl,(0,0)}}}
}
\newcommand\LatitudePlane[3][current plane]{%
  \pgfmathsinandcos\sinEl\cosEl{#2} 
  \pgfmathsinandcos\sint\cost{#3} 
  \pgfmathsetmacro\yshift{\cosEl*\sint}
  \tikzset{#1/.estyle={cm={\cost,0,0,\cost*\sinEl,(0,\yshift)}}} %
}
\newcommand\DrawLongitudeCircle[2][1]{
  \LongitudePlane{\angEl}{#2}
  \tikzset{current plane/.prefix style={scale=#1}}
  \pgfmathsetmacro\angVis{atan(sin(#2)*cos(\angEl)/sin(\angEl))} %
  \draw[current plane] (\angVis:1) arc (\angVis:\angVis+180:1); 
  \draw[current plane,dashed] (\angVis-180:1) arc (\angVis-180:\angVis:1); 
}
\newcommand\DrawLatitudeCircle[2][1]{
  \LatitudePlane{\angEl}{#2}
  \tikzset{current plane/.prefix style={scale=#1}}
  \pgfmathsetmacro\sinVis{sin(#2)/cos(#2)*sin(\angEl)/cos(\angEl)} 
  \pgfmathsetmacro\angVis{asin(min(1,max(\sinVis,-1)))}  
  \draw[current plane] (\angVis:1) arc (\angVis:-\angVis-180:1); 
  \draw[current plane,dashed] (180-\angVis:1) arc (180-\angVis:\angVis:1);
}

\newcommand\projection[5]{ 
    \pgfmathsinandcos\sinEl\cosEl{#1}
    \pgfmathsinandcos\sinlon\coslon{#2}
    \pgfmathsinandcos\sinlat\coslat{#3}
    \pgfmathsetmacro#4{\coslon*\coslat}
    \pgfmathsetmacro#5{\sinlon*\sinEl*\coslat+\sinlat*\cosEl}
}

\newcommand\DrawArcNoMark[9][\R]{
    \pgfmathsinandcos\sinEl\cosEl{#2}
    \pgfmathsinandcos\sinxlon\cosxlon{#3}
    \pgfmathsinandcos\sinxlat\cosxlat{#4}
    \pgfmathsinandcos\sinylon\cosylon{#5}
    \pgfmathsinandcos\sinylat\cosylat{#6}
    \pgfmathsinandcos\sinxylat\cosxylat{#4+#6}
    
    \pgfmathsetmacro\arclen{acos(\sinxlat*\sinxylat+\cosxlat*\cosxylat*\cosylon)}
    \pgfmathsinandcos\sinarclen\cosarclen{\arclen}
    \pgfmathsetmacro\zlon{\ifnumcomp{0}{=}{#4}{#3+90}{atan(-\sinylon/(\sinxlat*(\cosxlat*\sinxylat/\cosxylat-\sinxlat*\cosylon)))}}
    \ifnumcomp{0}{>}{\zlon}{\pgfmathsetmacro\zlon{\zlon+180}}{}
    
    \pgfmathsinandcos\sinzlon\coszlon{\zlon}
    
    \pgfmathsetmacro\zlat{atan(\coszlon*\sinxlat/\cosxlat+\sinzlon*(\sinxylat/\cosxylat-\cosylon*\sinxlat/\cosxlat)/\sinylon)}
    \pgfmathsinandcos\sinzlat\coszlat{\zlat}
    \projection\angEl{\zlon+#3}\zlat\zx\zy
    \projection\angEl{#3}{#4}\startx\starty

    \tikzset{newplane/.estyle={cm={
        \startx,
        \starty,
        \zx,
        \zy,
        (0,0)
    },scale=#1}}
    \draw[newplane,->,line width=0.35mm] (1,0) arc (0:\arclen/2:1) node[below right] {\small #7} arc (\arclen/2:\arclen:1) node[#9 right] {\small #8};
}

\tikzset{%
  >=latex, 
  inner sep=0pt,%
  outer sep=2pt,%
  mark coordinate/.style={inner sep=0pt,outer sep=0pt,minimum size=3pt,fill=black,circle}%
} 

\renewenvironment{proof}[1][\unskip]{\par\noindent{\bf Proof #1\ }}{\hfill\BlackBox\\[2mm]}

\renewcommand*{\vec}{\bm}
\newcommand{\dCC}[0]{d_{\rm \scriptscriptstyle CC}}

\newcommand{\bigO}[0]{\mathcal{O}}

\usepackage{array}
\newcolumntype{L}[1]{>{\raggedright\arraybackslash}p{#1}}

\ShortHeadings{The hyperspherical geometry of community detection}{G\"{o}sgens, van der Hofstad and Litvak}
\firstpageno{1}

\begin{document}

\title{The Hyperspherical Geometry of Community Detection:\\
Modularity as a Distance}

\author{\name Martijn G\"{o}sgens \email research@martijngosgens.nl \\
       \addr Eindhoven University of Technology, 
       Eindhoven, Netherlands
       \AND
       \name Remco van der Hofstad \email r.w.v.d.hofstad@tue.nl \\
       \addr Eindhoven University of Technology,
       Eindhoven, Netherlands
       \AND
       \name Nelly Litvak \email n.litvak@utwente.nl \\
       \addr University of Twente, 
       Enschede, Netherlands\\
       \addr Eindhoven University of Technology, 
       Eindhoven, Netherlands
       }


\maketitle

\begin{abstract}
We introduce a metric space of clusterings, where clusterings are described by a binary vector indexed by the vertex-pairs. We extend this geometry to a hypersphere and prove that maximizing modularity is equivalent to minimizing the angular distance to some \emph{modularity vector} over the set of clustering vectors. 
In that sense, modularity-based community detection methods can be seen as a subclass of a more general class of \emph{projection methods}, which we define as the community detection methods that adhere to the following two-step procedure: first, mapping the network to a point on the hypersphere; second, projecting this point to the set of clustering vectors.
We show that this class of projection methods contains many interesting community detection methods. Many of these new methods cannot be described in terms of null models and resolution parameters, as is customary for modularity-based methods.
We provide a new characterization of such methods in terms of \emph{meridians} and \emph{latitudes} of the hypersphere.
In addition, by relating the modularity resolution parameter to the latitude of the corresponding modularity vector, we obtain a new interpretation of the \emph{resolution limit} that modularity maximization is known to suffer from.
\end{abstract}
\begin{keywords}
    Community detection, Louvain algorithm, modularity, clustering, resolution limit
\end{keywords}

\section{Introduction}\label{sec:introduction}
Complex networks often contain groups of nodes that are more connected internally than externally. In network science, such groups are often referred to as \emph{communities}. Community detection is the task of detecting such groups in a network. Many algorithms have been developed for this task and its variants ~\citep{fortunato2010community,fortunato2016community,rosvall2019different}. In this work, we consider detecting non-overlapping communities, so that a community structure is represented by a partition of the network nodes into communities. We refer to such partitions as \emph{clusterings}. We define a hyperspherical geometry on such clusterings that allow for a natural interpretation of widely-used community detection methods in terms of projections. This geometrical interpretation has several interesting consequences, and opens the door for designing new community detection methods.

In the remainder of this section we explain our main contributions and innovations. In order to do so, we start by introducing basic notations and key notions from the literature that we build upon in this work. 
We consider a simple undirected graph $G$ with vertex-set $[n]=\{1,\dots,n\}$. For a vertex $i\in[n]$, we denote its degree by $d_i^{(G)}$ and denote the total number of edges by $m_G=\tfrac{1}{2}\sum_{i\in[n]}d_i^{(G)}$. We use the letters $C$ and $T$ to denote clusterings. 

\paragraph{Granularity.} We denote the number of clusters that a clustering consists of by $|C|$, while $m_C$ denotes the number of intra-cluster vertex-pairs, which is an integer in between 0 (when $|C|=n$) and the total number of vertex-pairs $N={n\choose2}$ (when $|C|=1$). The extent to which a clustering resembles the former or latter of these extremes is often referred to as the \emph{granularity} of a clustering: \emph{fine-grained} clusterings consist of many small clusters while \emph{coarse-grained} clusterings consist of a few large clusters. Granularity is an important notion in this paper, and we use the number of intra-cluster pairs $m_C$ as a measure of the granularity. Alternatively, some works simply quantify granularity by the number of clusters $|C|$, but this has the drawback of being insensitive to the cluster sizes. In other works~\citep{romano2016adjusting,vinh2009information,vinh2010information}, Shannon entropy is used as a measure of granularity. A generalization of Shannon entropy has been shown to be related to $m_C$~\citep{romano2016adjusting}.

\paragraph{Validation indices.} When the ground truth community structure is available, the performance of a community detection method can be evaluated by comparing the obtained candidate clustering to the ground truth. In this work, we denote the ground truth clustering by $T$ while we denote other (candidate) clusterings by $C$.  Functions that quantify similarity of $C$ and $T$ are referred to as \emph{validation indices}. There exist many different validation indices~\citep{vinh2010information,Lei2017} and which one is most suitable depends on the context of the application~\citep{gosgens2021systematic}.

One popular class of validation indices are \emph{pair-counting functions}. Given clusterings $T$ and $C$, these indices can be expressed as functions of the following four variables: the number of intra-cluster pairs $m_T,m_C$ of $T$ and $C$, respectively; the number of vertex-pairs $m_{TC}$ that are intra-cluster pairs of both $T$ and $C$; and the total number of vertex-pairs $N$. Examples of such pair-counting indices are the Rand index~\citep{rand1971objective}, the Jaccard index~\citep{jaccard1912distribution}, the Hubert Index~\citep{Hubert1977} and the Correlation Coefficient~\citep{gosgens2021systematic}.
Many of these indices are known to have a bias towards either clusterings of fine or coarse granularity~\citep{albatineh2006similarity,romano2016adjusting,gosgens2021systematic,Lei2017}. The Jaccard index is known to be biased towards coarse-grained clusterings while the bias of Rand depends on the granularity of the ground truth~\citep{Lei2017}. The Correlation Coefficient does not suffer from such bias and additionally satisfies many other desirable properties~\citep{gosgens2021systematic}.
This correlation coefficient is given by
\begin{equation}\label{eq:CC}
\textrm{CC}(T,C)=\frac{m_{TC}N-m_Tm_C}{\sqrt{m_T\cdot(N-m_T)\cdot m_C\cdot(N-m_C)}}.
\end{equation}
One of the desirable properties of the correlation coefficient that we make grateful use of is that it can be transformed to a metric distance by taking its arccosine~\citep{gosgens2021systematic}.
This \emph{Correlation Distance} plays a central role in this work and is used as the main validation index throughout this paper.

\paragraph{Modularity.} One of the most popular ways to detect communities is by maximizing a quantity called \emph{modularity}~\citep{newman2004finding,reichardt2006statistical}.
Modularity is based on the paradigm that a graph with community structure has many more intra-communities edges than it would have if one were to rewire the edges at random.

Given a clustering $C$ and a graph $G$, modularity measures the difference between the number of intra-cluster edges that are present in $G$ and the expected number of intra-cluster edges if $G$ were to be rewired according to some random graph model without community structure. Such a random graph model is usually referred to as a \emph{null model}. While it is theoretically possible to use any random graph model as a null model, the null models that are commonly used in the literature are the Erd\H{o}s-R\'{e}nyi (ER) model and Configuration Model (CM).
The main result of this paper in Theorem~\ref{thm:modularity_distance} establishes the equivalence between modularity maximization and a nearest neighbor search in hyperspherical geometry.

\paragraph{The resolution limit.} In large graphs, modularity maximization is known to be unable to detect communities below a given size. This \emph{resolution limit}~\citep{Fortunato2007,kumpula2007limited,lancichinetti2011limits} is a serious drawback of modularity-based methods. Basically, it means that modularity maximization is implicitly tuned to detect clusterings of a certain granularity. Modularity is often extended to include a \emph{resolution parameter} to alleviate this problem~\citep{reichardt2006statistical,kumpula2007limited}. However, this resolution parameter merely allows one to tune the detection method to a different granularity~\citep{arenas2008analysis}, and does not address the fundamental problem that modularity implicitly has a bias towards clusterings of some granularity~\citep{kumpula2007limited,traag2011narrow}.
Furthermore, it is nontrivial to find a resolution parameter value so that modularity optimization is tuned to the desired granularity in a given setting~\citep{arenas2008analysis,prokhorenkova2019using}. Recently, modularity optimization has been shown to be equivalent to likelihood maximization~\citep{newman2016equivalence}:
for a particular value of the resolution parameter, ER modularity is equivalent to the likelihood function of a Planted Partition Model while CM modularity is equivalent to the likelihood of a Degree-Corrected Stochastic Block Model. While this provides a mathematically-principled approach to choosing the resolution parameter~\citep{prokhorenkova2019community}, it does require making assumptions about the distribution that the network was drawn from and knowledge of its parameters.

\paragraph{The Louvain algorithm.} Despite its shortcomings, modularity maximization remains one of the most popular approaches for community detection. In applications, an additional difficulty arises from the fact that modularity maximization is NP-hard~\citep{brandes2007modularity}, which makes its exact maximization infeasible for large graphs. Therefore, practitioners often resort to \emph{approximate} optimization algorithms, the most popular being the so-called Louvain algorithm~\citep{blondel2008fast}. This algorithm is known to find an approximate modularity maximum in roughly log-linear time.
The Louvain algorithm is also known to efficiently optimize other partition-based functions~\citep{prokhorenkova2019community}.
Recently, some improvements to the Louvain algorithm have been made, resulting in algorithms of similar running time that reach better local maxima~\citep{traag2019louvain}.
However, these improved algorithms are generally not as easily extended to other partition-based functions. For this reason, we limit ourselves to the Louvain algorithm for the remainder of the paper.

\paragraph{Main innovations in this paper.}
The central idea of this paper is to describe a {\em hyperspherical geometry} on the set of clusterings. Our main result is that in terms of this geometry, maximizing modularity is equivalent to minimizing the angular distance to some point on the hypersphere that we call the {\em modularity vector}, over the set of clusterings.
In this geometric viewpoint, modularity can be seen as a class of mappings (parametrized by the null model and resolution parameter) from the set of networks to points on the hypersphere.
Then, any algorithm that maximizes modularity (e.g., the well-known Louvain algorithm) can be seen as an (approximate) nearest-neighbor search, finding the clustering that minimizes the distance to the modularity vector.
By allowing for different ways of mapping networks to points on the hypersphere, we obtain a general class of community detection methods that we refer to as \emph{projection methods}, as the network is first mapped to a point on the hypersphere and then projected to a clustering vector.
We show that many interesting projection methods exist that lie outside the class of modularity-based methods.

\paragraph{Data clustering.}
Note that community detection is closely related to the more general field of data clustering, as we are essentially clustering the network nodes based on network topology.
While the focus of the present paper is community detection in networks, the proposed method is not limited to networks alone and is able to cluster \emph{any} pair-wise similarity data, such as the upper or lower triangle of any affinity matrix.
For a more general overview of data clustering, we refer to~\cite{jain2010data}.

\paragraph{Organisation of this paper.}
This paper is organised as follows.
In Section~\ref{sec:clustering_geometry}, we describe the hyperspherical geometry on the set of clusters. In Section~\ref{sec:modularity_distance}, we show that maximizing modularity is equivalent to minimizing the angular distance to a modularity vector. Section~\ref{sec:consequences} then discusses the consequences of this equivalence for modularity-based methods, while Section~\ref{sec:query_mappings} explores the projection methods that lie outside the class of modularity-based methods. Finally, in Section~\ref{sec:experiments}, we compare a number of projection methods on real-world networks and interpret the results in our geometric framework.

\section{The hyperspherical geometry of clustering}\label{sec:clustering_geometry}
For a given clustering $C$, we define $\vec{b}(C)$ as the binary $N$-dimensional vector indexed by the vertex-pairs, where $\vec{b}(C)_{ij}$ equals $+1$ if $i$ and $j$ are assigned the same cluster in $C$ and $-1$ otherwise.
In this work, binary vectors have entries $+1$ and $-1$. In some particular cases, we need binary vectors with values $1$ and $0$, which are obtained by the transformation $\tfrac{1}{2}(\vec{b}+\vec{1})$, where $\vec{1}$ denotes the vector where all entries equal $1$.
Note that not all $\pm1$ binary vectors correspond to clusterings, since the vector needs to satisfy transitivity. That is, for vertices $i,j,k$, it must hold that if $\vec{b}(C)_{ij}=\vec{b}(C)_{jk}=1$ then $\vec{b}(C)_{ik}=1$. Furthermore, every binary vector that does satisfy this transitivity condition corresponds to precisely one clustering $C$.
Importantly, note that all binary vectors have equal (Euclidean) length $\sqrt{N}$, hence, all clustering vectors lie on a hypersphere with radius $\sqrt{N}$ centered around the origin.

For clusterings $C_1,C_2$, the \emph{angular distance} between their binary vectors is given by
\begin{equation}\label{eq:angular_distance}
d_a(\vec{b}(C_1),\vec{b}(C_2))=\arccos\left(\frac{\langle\vec{b}(C_1),\vec{b}(C_2)\rangle}{N}\right),
\end{equation}
where $\langle\cdot,\cdot\rangle$ denotes the standard inner product.
We can easily extend this metric space to the full hypersphere to allow for non-binary vectors. For  $\vec{x},\vec{y}\in\mathbb{R}^N$, the angular distance is
\[
d_a(\vec{x},\vec{y})=\arccos\left(\frac{\langle\vec{x},\vec{y}\rangle}{\|\vec{x}\|\cdot\|\vec{y}\|}\right),
\]
where $\|\vec{x}\|=\sqrt{\langle\vec{x},\vec{x}\rangle}$ denotes the Euclidean norm.
Furthermore, the correlation distance, which is defined as the arccosine of the correlation coefficient as given in~\eqref{eq:CC}, can similarly be extended to the full hypersphere by
\begin{equation}\label{eq:CD_full}
    \dCC(\vec{x},\vec{y})=\arccos\left(\frac{\langle\vec{x},\vec{y}\rangle-\langle\vec{x},\vec{1}\rangle\cdot\langle\vec{y},\vec{1}\rangle/N}{\sqrt{\|\vec{x}\|^2-\langle\vec{x},\vec{1}\rangle^2/N}\sqrt{\|\vec{y}\|^2-\langle\vec{y},\vec{1}\rangle^2/N}}\right).
\end{equation}
We note that the length of a vector has no meaning in these metric spaces. That is, scaling a vector by a positive scalar does not affect the (angular or correlation) distance to any other vector. This introduces an equivalence relation among vectors, where each equivalence class corresponds to a direction, the representative element of each of these classes is the vector that is on the surface of the hypersphere, given by the \emph{hypersphere projection}
\[
\mathcal{H}(\vec{x})=\frac{\sqrt{N}}{\|\vec{x}\|}\vec{x}.
\]
All the definitions in the remainder of this section are invariant under this hypersphere projection.

\paragraph{Hyperspherical geometry.} The angular distance defines a geometry on our hypersphere. We now introduce some basic hyperspherical geometry theory that will be needed later on.
These hyperspherical results are direct analogues of their spherical counterparts. We refer to \cite{todhunter1863spherical} and \cite{donnay2011spherical} for a more complete overview of spherical geometry and trigonometry.

\begin{figure}
    \centering
    \resizebox{0.7\textwidth}{!}{%
        \begin{tikzpicture} 


\def\R{2.5} 
\def\angEl{25} 
\def\angAz{-115} 
\def\zlong{75}
\def\zlat{30}
\def\xlat{0}


\pgfmathsetmacro\H{\R*cos(\angEl)} 
\tikzset{xyplane/.style={cm={cos(\angAz),sin(\angAz)*sin(\angEl),-sin(\angAz),
                              cos(\angAz)*sin(\angEl),(0,-\H)}}}
\LongitudePlane[xzplane]{\angEl}{\angAz}
\LongitudePlane[pzplane]{\angEl}{\angAz+\zlong}
\LatitudePlane[equator]{\angEl}{0}


\draw[xyplane] (-2*\R,-2*\R) rectangle (2*\R,2*\R);
\node[anchor=east] at (-\R,-1.2*\R) {Tangent plane of $\vec{y}$};
\fill[ball color=white] (0,0) circle (\R); 
\draw (0,0) circle (\R);


\coordinate[mark coordinate] (O) at (0,0);
\coordinate[mark coordinate] (N) at (0,\H);
\coordinate[mark coordinate] (S) at (0,-\H);
\path[pzplane] (\zlat:\R) coordinate[mark coordinate] (z);
\path[xzplane] (\xlat:\R) coordinate[mark coordinate] (x);
\pgfmathsetmacro\cosz{cos(\zlat)};
\pgfmathsetmacro\cosx{cos(\xlat)};
\path[xyplane] (\zlong:\R*\cosz) coordinate[mark coordinate] (zhat);
\path[xyplane] (0:\R*\cosx) coordinate[mark coordinate] (xhat);

\pgfmathsetmacro\perplength{\R/7};
\draw[xzplane,thin] (0,-\R) -- (0,\perplength-\R) -- (\perplength,\perplength-\R) -- (\perplength,-\R) -- cycle;
\draw[pzplane,thin] (0,-\R) -- (0,\perplength-\R) -- (\perplength,\perplength-\R) -- (\perplength,-\R) -- cycle;
\draw[pzplane] (zhat) -- ($ (zhat)+ (-\perplength,0) $) -- ($ (zhat)+ (-\perplength,\perplength) $) -- ($ (zhat)+ (0,\perplength) $);
\draw[xzplane] (xhat) -- ($ (xhat)+ (-\perplength,0) $) -- ($ (xhat)+ (-\perplength,\perplength) $) -- ($ (xhat)+ (0,\perplength) $);


\DrawLatitudeCircle[\R]{\xlat} 
\DrawLatitudeCircle[\R]{\zlat}
\DrawLongitudeCircle[\R]{\angAz} 
\DrawLongitudeCircle[\R]{\angAz+\zlong} 


\draw[xyplane,<->] (0:1.5*\R) -- (0,0) -- (\zlong:1.5*\R);
\draw (0,-\H) -- (0,0) node[below right] {$\vec{0}$};


\draw[->] (z)-- (zhat) node[above right] {$\vec{\hat{z}}$};
\draw[->] (x) -- (xhat) node[above left] {$\vec{\hat{x}}$};
\path (S) +(0.4ex,-0.4ex) node[below] {$\vec{y}$};
\draw[->] (O) -- (z) node[above left] {$\vec{z}$};
\draw[->] (O) -- (x) node[below left] {$\vec{x}$};
\draw[xyplane,<->,thin] (1.3*\R,0) arc (0:\zlong/3:1.3*\R) node[anchor=north]{$\angle(\vec{x},\vec{y},\vec{z})$} arc (\zlong/3:\zlong:1.3*\R);

\end{tikzpicture}
    }
    \caption{Illustration of the hyperspherical angle $\angle(\vec{x},\vec{y},\vec{z})$ in the three-dimensional subspace spanned by $\vec{x},\vec{y},\vec{z}$. The solid arcs are on the visible side of the sphere while the dashed arcs are on the back side. The vectors $\hat{\vec{x}}$ and $\hat{\vec{z}}$ are the projections of $\vec{x},\vec{z}$ respectively to the tangent plane of $\vec{y}$. For $\vec{y}=-\vec{1}$, we have $\angle(\vec{x},-\vec{1},\vec{z})=\dCC(\vec{x},\vec{z})$.}
    \label{fig:hyperspherical_angle}
\end{figure}
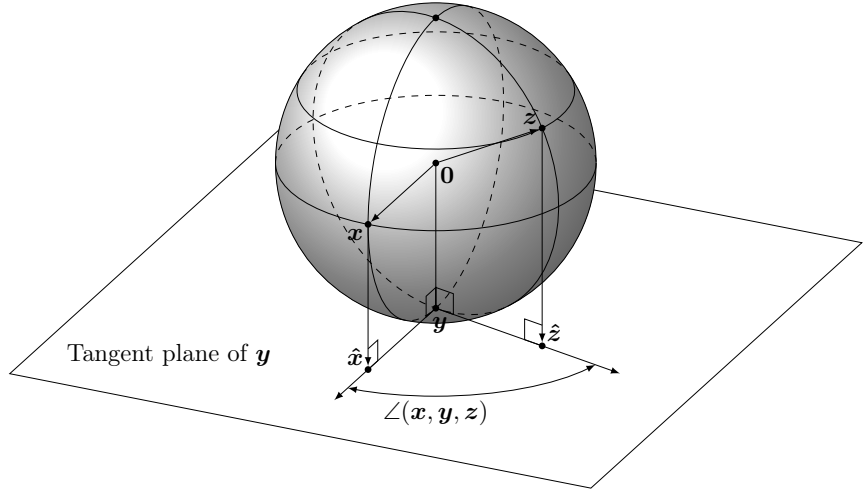

For any two-dimensional plane that contains the origin, its intersection with the hypersphere corresponds to a \emph{great circle} (e.g., the equator of a globe), which is the closest hyperspherical analogue to an infinite straight line in Euclidean geometry. Therefore, we refer to segments of a great circle as \emph{hyperspherical straight lines}. The length of a hyperspherical straight line between two points corresponds to the angular distance between its endpoints.
The definition of a hyperspherical angle (the angle that two hyperspherical straight lines make on the hypersphere) is less straightforward:
given three distinct points $\vec{x},\vec{y},\vec{z}$ on the hypersphere, the hyperspherical angle $\angle(\vec{x},\vec{y},\vec{z})$ that the line from $\vec{x}$ to $\vec{y}$ makes with the line from $\vec{z}$ to $\vec{y}$ at $\vec{y}$ is given by the \emph{hyperspherical cosine rule}:
\begin{equation}\label{eq:hypercosine}
    \cos\angle(\vec{x},\vec{y},\vec{z})=\frac{\cos d_a(\vec{x},\vec{z})-\cos d_a(\vec{x},\vec{y})\cos d_a(\vec{y},\vec{z})}{\sin d_a(\vec{x},\vec{y})\sin d_a(\vec{y},\vec{z})}.
\end{equation}
This angle is obtained by projecting $\vec{x}$ and $\vec{z}$ on the tangent plane of $\vec{y}$ on the hypersphere and then simply computing the angle for these projected points via the (Euclidean) cosine rule, as illustrated in Figure~\ref{fig:hyperspherical_angle}. 
Later in this section, we show that the Correlation Distance is a hyperspherical angle.

\paragraph{Poles.} In the metric space described previously, the all-one vector $\vec{1}$ corresponds to the clustering where all items are put in the same cluster, while the vector $-\vec{1}$ corresponds to the clustering where each item is placed in its own separate cluster. These two points form opposite \emph{poles} on our hypersphere. Clusterings lying close to $\vec{1}$ correspond to clusterings of coarse granularity (i.e., consisting of relatively few and large clusters), while clusters close to $-\vec{1}$ correspond to clusterings of fine granularity. Therefore, we refer to $\vec{1}$ as the \emph{coarse pole} while we refer to $-\vec{1}$ as the \emph{fine pole}.

\paragraph{Latitude and granularity.} In analogy to the terminology used for globes, we refer to the angular distance to the fine pole as the \emph{latitude} of a vector\footnote{Note that this slightly differs from the definition used in geography, where the latitude is the signed angular distance to the equator.}. We thus define the latitude of a vector $\vec{x}$ by
\begin{equation}\label{eq:latitude}
\ell(\vec{x})=d_a(\vec{x},-\vec{1})=\arccos\Big(\frac{-\langle\vec{x},\vec{1}\rangle}{\sqrt{N}\cdot\|\vec{x}\|}\Big).
\end{equation}
If $\vec{x}$ is a binary vector, then \eqref{eq:latitude} depends only on the number of $+1$'s in $\vec{x}$. Therefore, for a clustering $C$ with $m_C$ intra-cluster pairs, the latitude is given by
\begin{equation}\label{eq:clustering_lat}
\ell(\vec{b}(C))=\arccos\left(\frac{-m_C+(N-m_C)}{\sqrt{N}\cdot \sqrt{N}}\right)=\arccos\left(1-\frac{2m_C}{N}\right).
\end{equation}
As mentioned previously, $m_C$ is a measure of the granularity of the clustering $C$. Since $\ell(\vec{b}(C))$ is a monotonous transformation of $m_C$, it can equivalently be viewed as a measure of granularity.

We refer to the \emph{equator} as the set of points at equal distance to both poles, that is, with latitude $\pi/2$, or equivalently, the set of vectors $\vec{x}$ with $\langle\vec{x},\vec{1}\rangle=0$. This contains the clustering vectors for clusterings that have an equal number of inter- and intra-cluster pairs.
From \eqref{eq:clustering_lat} we see that $\ell(\vec{b}(C))>\pi/2$ occurs only when $C$ has more intra-cluster pairs than inter-cluster pairs. Moreover, it can be shown that any clustering with $\ell(\vec{b}(C))>\pi/2$ has a cluster of size larger than $n/2$,  which is usually not the case when clustering into more than two clusters. This tells us that clusterings generally lie on the fine hemisphere.
Furthermore, let $n\rightarrow\infty$ and suppose that the expected community size of a randomly chosen vertex is $s$ and does not depend on $n$, then the latitude of a ground truth community structure shrinks like         
    \begin{equation}\label{eq:clustering_concentration}
    \arccos\left(1-2\frac{s-1}{n-1}\right)=2\sqrt{\frac{s-1}{n-1}}+o(n^{-1}),
    \end{equation}
so that for large $n$, clusterings are concentrated around the fine pole.

\paragraph{Meridians and parallels.} In the same geographic analogy as above, a \emph{meridian} is a hyperspherical straight line between the poles, while a \emph{parallel} is a hypersurface of constant latitude. Note, however, that in the three-dimensional real-world, a parallel is a one-dimensional object while in our case it has dimension $N-2$. 

For every vector $\vec{x}$ that is not a multiple of $\vec{1}$, there is precisely one meridian running through its hypersphere projection. Similarly, every vector $\vec{x}$ lies on exactly one parallel, namely the parallel with latitude $\ell(\vec{x})$.
Therefore, each point on the hypersphere is uniquely defined by a meridian and a latitude. The point that lies on the meridian of $\vec{x}$ and has latitude $\lambda$ is equal to the projection of $\vec{x}$ onto the parallel with latitude $\lambda$. This \emph{parallel projection} is given by
\begin{equation}\label{eq:parallel_projection}
\mathcal{P}_\lambda(\vec{x})=\sin(\lambda)\mathcal{H}\left(\vec{x}-\frac{\langle\vec{x},\vec{1}\rangle}{N}\vec{1}\right)-\cos(\lambda)\vec{1}.
\end{equation}
Since the equator corresponds to the parallel with latitude $\pi/2$, the projection to the equator is given by $\mathcal{P}_{\pi/2}(\vec{x})$.

All meridians meet at both poles. Given two meridians that run through $\vec{x}$ and $\vec{y}$ respectively, we can compute the hyperspherical angle that they make at the fine pole. We refer to this as the \emph{meridian angle} between $\vec{x}$ and $\vec{y}$. For $\vec{x},\vec{y}$ with latitudes in $(0,\pi)$, the meridian angle can be computed via the hyperspherical cosine rule~\eqref{eq:hypercosine} and is given by
\begin{equation}\label{eq:meridian_hypercosine}
    \textrm{MeridianAngle}(\vec{x},\vec{y})=\angle(\vec{x},-\vec{1},\vec{y})=
    \arccos\left(\frac{\cos d_a(\vec{x},\vec{y})-\cos\ell(\vec{x})\cos\ell(\vec{y})}{\sin\ell(\vec{x})\sin\ell(\vec{y})}\right).
\end{equation}
Alternatively, this meridian angle can be written as $d_a(\mathcal{P}_{\pi/2}(\vec{x}),\mathcal{P}_{\pi/2}(\vec{y}))$.
By convention, we define this angle to be zero when both vectors are on the same pole, $\pi$ when both vectors are on opposite poles and $\pi/2$ when one vector is on a pole while the other is not.
We prove the following result:
\begin{theorem}\label{thm:meridian_correlation}
    For all $\vec{x},\vec{y}\in\mathbb{R}^N$, we have ${\rm MeridianAngle}(\vec{x},\vec{y})=\dCC(\vec{x},\vec{y})$. That is, the meridian angle coincides with the correlation distance.
\end{theorem}
\begin{proof}
 We prove the theorem by showing that the cosine of the meridian angle in \eqref{eq:meridian_hypercosine} and the cosine of the correlation distance \eqref{eq:CD_full} are the same. From \eqref{eq:meridian_hypercosine} we directly get
 \[
 \cos(\textrm{MeridianAngle}(\vec{x},\vec{y}))=\frac{\cos d_a(\vec{x},\vec{y})-\cos\ell(\vec{x})\cdot\cos\ell(\vec{y})}{\sqrt{1-\cos^2\ell(\vec{x})}\cdot \sqrt{1-\cos^2\ell(\vec{y})}}.
 \]
 Next, we multiply both numerator and denominator by $\|\vec{x}\|\cdot\|\vec{y}\|$ to obtain
 \begin{align*}
 \frac{\|\vec{x}\|\cdot\|\vec{y}\|\cos d_a(\vec{x},\vec{y})-\|\vec{x}\|\cos\ell(\vec{x})\cdot \|\vec{y}\|\cos\ell(\vec{y})}{\sqrt{\|\vec{x}\|^2-(\|\vec{x}\|\cos\ell(\vec{x}))^2}\cdot \sqrt{\|\vec{y}\|^2-(\|\vec{y}\|\cos\ell(\vec{y}))^2}}.
 \end{align*}
 After substituting $\|\vec{x}\|\cdot\|\vec{y}\|\cos d_a(\vec{x},\vec{y})=\langle\vec{x},\vec{y}\rangle$, $\|\vec{x}\|\cos\ell(\vec{x})=-\langle\vec{x},\vec{1}\rangle/\sqrt{N}$ (see \eqref{eq:latitude}) and similarly $\|\vec{y}\|\cos\ell(\vec{y})=-\langle\vec{y},\vec{1}\rangle/\sqrt{N}$, the resulting expression is the same as the argument of the arccosine in the right-hand side of  \eqref{eq:CD_full}.
 
 We next check the boundary cases, where $\vec{x}$ or $\vec{y}$ are poles:
 The case where both $\vec{x}$ and $\vec{y}$ are on the same pole corresponds to comparing two constant vectors of the same sign, which gives correlation $1$ and thus results in a correlation distance of $0$, which is equal to the meridian angle by our definition. 
 Similarly, when both vectors are on different poles, they are constant vectors of different signs, hence, the correlation coefficient is $-1$, and the correlation distance again equals the meridian angle $\pi$.
 Finally, when one vector is on the pole and the other is not,  \eqref{eq:CD_full} gives $\dCC(\vec{x},\vec{y})=\pi/2$, again in accordance with our definition of the meridian angle.
\end{proof}

Theorem~\ref{thm:meridian_correlation} provides a new geometric interpretation for the correlation distance.
Furthermore, since the correlation distance defines a distance on the set of clustering vectors, it can be zero only when clusterings are equal. This implies that all clustering vectors, except the poles, lie on different meridians.
To avoid using two different notations for the same quantity, from now on we denote the meridian angle by $\dCC$ and refer to it as the correlation distance in the remainder of this paper. As mentioned in Section~\ref{sec:introduction}, we use the correlation distance as validation measure in this paper. That is, if $T$ is the ground truth clustering of our network while $C$ is the clustering obtained by some community detection method, then we quantify the performance of that method on that network by $\dCC(\vec{b}(T),\vec{b}(C))$.

\paragraph{Other pair-counting distance metrics.}
The correlation distance is not the only pair-counting validation index that is related to this hyperspherical geometry:
the cosine of the angular distance $d_a$ is equal to the Hubert index~\citep{Hubert1977}. This Hubert index is related to the Rand index by $\textrm{Hubert}(C_1,C_2)=2\textrm{Rand}(C_1,C_2)-1$, while the Rand index is related to the Euclidean distance $d_E$ by 
\[
  \textrm{Rand}(C_1,C_2)=1-\frac{d_E(\vec{b}(C_1),\vec{b}(C_2))^2}{4N}.
\]
These two relations are relevant because modularity can be related to both Euclidean and angular distance, as we will prove in Section~\ref{sec:modularity_distance}.

\section{Modularity as a distance}\label{sec:modularity_distance}
In the previous section, we have shown how clusterings are mapped to binary vectors in $\mathbb{R}^N$ in a meaningful way by the mapping $\vec{b}(\cdot)$. In the present section, we discuss how graphs can be mapped to $\mathbb{R}^N$ in a similar way. This is the next step in our methodology that formalizes community detection methods as algorithms that find a candidate community structure $C$ by minimizing the distance between its binary vector $\vec{b}(C)$ and the \emph{query vector} $\vec{q}(G)$ that the graph is mapped to.  Here, the function $\vec{q}(\cdot)$ maps graphs of size $n$ to vectors in $\mathbb{R}^N$, and we refer to $\vec{q}(\cdot)$ as a \emph{query mapping}. Among others, we will prove that this setup generalizes modularity-based community detection methods. 

\paragraph{Query vectors versus affinity matrices.}
When performing community detection, we want to find $\vec{q}(G)$ such that the nearest clustering vector is close to the `true' clustering $T$. Therefore, ideally, $\vec{q}(G)$ should be chosen in such a way that $\vec{b}(T)$ is one of its nearest clustering vectors. To achieve this, we wish $\vec{q}(G)_{ij}$ to be positive whenever $ij$ is an intra-community pair of $T$ and negative otherwise.
In that sense, query vectors are related to the notion of an \emph{affinity matrix} that is used throughout the clustering literature~\citep{xu2005survey,filippone2008survey}, where the pairwise similarities are summarized by an $n\times n$ symmetric matrix.
Existing clustering algorithms use this matrix in various ways to cluster the elements. For example, the popular Spectral Clustering algorithm computes the leading eigenvectors of the affinity matrix and then clusters the items based on these vectors~\citep{filippone2008survey}.
In our context, the ${n\choose 2}=N$ entries of a pair-wise vector correspond to the entries of the upper (or lower) triangle of this affinity matrix.
While for many operations, such as computing eigenvectors, a matrix representation of similarity is natural, we find that in our geometric setting, a vector representation is more convenient because we heavily rely on inner products.

The simplest way to map a graph $G$ to a vector is by a binary \emph{edge-connectivity vector} $\vec{e}(G)$, where $\vec{e}(G)_{ij}$ equals $+1$ if vertices $i$ and $j$ are connected by an edge and $-1$ otherwise. We note that this mapping $\vec{e}(\cdot)$ forms a bijection between the set of simple undirected graphs of $n$ vertices and the set of binary vectors.
However, there are many more ways in which a graph can be mapped to $\mathbb{R}^N$.
In Theorem~\ref{thm:modularity_distance} below we show that maximizing modularity $M^{\mathcal{N}}(C,G;\gamma)$, for any null-model $\mathcal{N}$ and resolution parameter $\gamma$,  is equivalent to minimizing the distance between the clustering vector $\vec{b}(C)$ and the \emph{modularity vector} $\vec{q}(G)=\vec{q}_M^{\mathcal{N}}(G;\gamma)$.
Let the $N$-dimensional vector $\vec{p}^{\mathcal{N}}(G)$ denote the expected number of edges that each vertex-pair of $G$ has when the edges are rewired according to a null-model $\mathcal{N}$. This vector is non-negative and its entries sum to the total number of edges $m_G$ as the rewiring keeps the total number of edges unchanged.
We then define the modularity vector with respect to the null model $\mathcal{N}$ and resolution parameter $\gamma$ as
\begin{equation}\label{eq:modularity_vector}
    \vec{q}_M^{\mathcal{N}}(G;\gamma)=\vec{1}+\vec{e}(G)-2\gamma \vec{p}^{\mathcal{N}}(G).
\end{equation}
When $\gamma=1$, we may sometimes write $\vec{q}_M^{\mathcal{N}}(G)$ instead of $\vec{q}_M^{\mathcal{N}}(G;1)$ for brevity.
For an Erd\H{o}s-R\'{e}nyi null model, the expected number of edges is given by $\vec{p}^{\textrm{ER}}(G)=\frac{m_G}{N}\vec{1}$, while for the Configuration Model, it is given by
\begin{equation}\label{eq:modularity_vec_ER_CM} 
\vec{p}^{\textrm{CM}}(G)_{ij}=\frac{d_i^{(G)}d_j^{(G)}}{2\widetilde{m_G}},
\end{equation}
where $\widetilde{m_G}$ is given by
\[
\widetilde{m_G}=m_G-\frac{\sum_{i\in[n]}\left(d_i^{(G)}\right)^2}{4m_G}.
\]
In other works, the denominator $2\widetilde{m_G}$ is usually replaced by $2m_G$ for simplicity. In our setting, we need to use $2\widetilde{m_G}$ to ensure that the expectations sum to $m_G$. Note that in the large-graph limit, both options are equivalent, since 
$\lim\limits_{n\rightarrow\infty}\tfrac{\widetilde{m_G}}{m_G}=1.$
We now present the main result of this paper:
\begin{theorem}\label{thm:modularity_distance}
    For a null model $\mathcal{N}$ and resolution parameter $\gamma$, maximizing modularity $M^{\mathcal{N}}(C,G;\gamma)$ over all clusterings $C$ is equivalent to minimizing either
\begin{enumerate}[label=(\roman*)]
    \item the Euclidean distance $d_E(\vec{b}(C),\vec{q}_M^{\mathcal{N}}(G;\gamma))$ over all clusterings $C$; or
    \item the angular distance $d_a(\vec{b}(C),\vec{q}_M^{\mathcal{N}}(G;\gamma))$ over all clusterings $C$.
\end{enumerate}
\end{theorem}

Theorem~\ref{thm:modularity_distance} thus gives a geometric interpretation of modularity maximization. We will see that this geometric interpretation is useful. For example, it paves the way for a generalization of modularity-based community detection methods and provides a new interpretation of the resolution parameter $\gamma$, as we explain in more detail in Section~\ref{sec:consequences}. 

To prove Theorem~\ref{thm:modularity_distance}, we rely on the following lemma:

\begin{lemma}\label{lem:euc2cad}
Minimizing the Euclidean distance to some $\vec{v}$ over the set of $\pm1$ binary vectors $\vec{b}$ is equivalent to minimizing the angular distance between $\vec{b}$ and $\vec{v}$. 
\end{lemma}

\begin{proof}
    The result follows after we relate the square of the Euclidean distance to the cosine of the angular distance, and then use the fact that any binary vector $\vec{b}$ has length $\sqrt{N}$: 
    \begin{align*}
        d_E(\vec{b},\vec{v})^2
        &=\|\vec{b}\|^2+\|\vec{v}\|^2-2\langle\vec{b},\vec{v}\rangle\\
        &=N+\|\vec{v}\|^2-2\|\vec{v}\|\sqrt{N}\cos d_a(\vec{b},\vec{v}).
    \end{align*}
    Now $N$ and $\|\vec{v}\|$ are constant w.r.t. $\vec{b}$ and can thus be ignored. This tells us that minimizing the Euclidean distance is equivalent to maximizing the cosine of the angular distance, or equivalently, minimizing the angular distance. This completes the proof.
\end{proof}

We are now ready to prove Theorem~\ref{thm:modularity_distance}:
\begin{proof}[of Theorem~\ref{thm:modularity_distance}]
    Modularity can be written as
    \begin{equation}\label{eq:modularity}
    M^{\mathcal{N}}(C,G;\gamma)=\frac{1}{m_G}\langle\tfrac{1}{2}(\vec{1}+\vec{b}(C)),\tfrac{1}{2}(\vec{1}+\vec{e}(G))-\gamma\vec{p}^{\mathcal{N}}(G)\rangle.
    \end{equation}
    We multiply this by $4m_G$ and subtract $\langle\vec{1},\vec{1}+\vec{e}(G)-2\gamma\vec{p}^{\mathcal{N}}(G)\rangle$, both of which are constant w.r.t.\ $C$. This yields
    \[
    \langle\vec{b}(C),\vec{1}+\vec{e}(G)-2\gamma\vec{p}^{\mathcal{N}}(G))\rangle=\|\vec{b}(C)\|\cdot\|\vec{q}_M^{\mathcal{N}}(G;\gamma)\|\cos d_a(\vec{b}(C),\vec{q}_M^{\mathcal{N}}(G;\gamma)).
    \]
    Finally, $\|\vec{b}(C)\|=\sqrt{N}$ and $\|\vec{q}_M^{\mathcal{N}}(G;\gamma)\|$ are constant w.r.t.\ $C$. From this, we conclude that maximizing modularity is equivalent to maximizing the cosine of the angular distance. Since the arccosine is a strictly decreasing function on $[-1,1]$, it follows that maximizing modularity is equivalent to minimizing the angular distance.
    The equivalence to the Euclidean distance follows immediately from Lemma~\ref{lem:euc2cad}.
\end{proof}

Theorem~\ref{thm:modularity_distance} relates modularity to a hyperspherical geometry as well as a Euclidean geometry, where the latter has one more dimension than the former. This can be explained by the fact that rescaling the query vector does not affect the optimization, as shown in Lemma~\ref{lem:euc2cad}. Therefore, the length of a query vector in Euclidean space has no effect on the resulting candidate clustering. Because of this, we focus on the hyperspherical geometry for the remainder of this work.

We note that Theorem~\ref{thm:modularity_distance} can easily be extended to variants of modularity. For example, modularity can be extended to weighted networks by replacing the positive entries of $\vec{e}(G)$ by the edge weights and replacing $d_i(G)$ by the sum of edge weights connected to vertex $i$. Similarly, the equivalence can also be extended to methods with negative edge-weights such as in \cite{traag2009community}. This shows that all modularity-based methods can be formulated in our geometrical framework. However, the class of methods that fit our geometrical framework is not limited to modularity-based methods, as we discuss in Section~\ref{sec:query_mappings}.

\section{Consequences for modularity-based methods}\label{sec:consequences}
In this section, we discuss several consequences of the equivalence proved in Theorem~\ref{thm:modularity_distance}. In particular, we provide a geometric interpretation for the Louvain algorithm and the resolution limit.

\paragraph{Louvain as an approximate nearest-neighbor search.}
Let us denote a clustering vector that minimizes the angular distance to the query vector $\vec{q}$ by $\mathcal{C}(\vec{q})$. Possibly, there are multiple clustering vectors at minimal distance to the query vector. In such cases, we choose $\mathcal{C}(\vec{q})$ arbitrarily (but deterministically) out of them.

Note that $\mathcal{C}(\vec{q})$ is a \emph{nearest neighbor} of the query vector $\vec{q}$ among the clustering vectors.
However, from Theorem~\ref{thm:modularity_distance} it follows that finding $\mathcal{C}(\vec{q}_M^{\mathcal{N}}(G;\gamma))$ is equivalent to finding the global modularity maximum, which is known to be NP-hard~\citep{brandes2007modularity}.
Therefore, most modularity-maximizing algorithms, such as the popular Louvain algorithm~\citep{blondel2008fast}, \emph{approximately} maximize modularity, resulting in an \emph{approximate nearest neighbor} of the modularity vector. Hence, the Louvain algorithm can be viewed as an \emph{approximate nearest-neighbor search}.

While the Louvain algorithm was initially introduced for  CM-modularity maximization, it is known to be able to efficiently optimize other clustering-functions as well~\citep{traag2011narrow,prokhorenkova2019community}. In this work, we use the Louvain algorithm to minimize the angular distance to a query vector.
Let $\mathcal{L}(\vec{q})$ denote the approximate nearest neighbor obtained by applying the Louvain algorithm to a query vector $\vec{q}$.
Then, we expect that $d_a(\mathcal{L}(\vec{q}),\vec{q})\approx d_a(\mathcal{C}(\vec{q}),\vec{q}),$
though $d_a(\mathcal{L}(\vec{q}),\vec{q})\geq d_a(\mathcal{C}(\vec{q}),\vec{q})$ always holds since $\mathcal{C}(\vec{q})$ is a global minimizer.

Furthermore, because the Louvain algorithm is a greedy algorithm that is initialized at the clustering corresponding to the fine pole $-\vec{1}$, it follows that $\mathcal{L}(\vec{q})$ cannot be further from $\vec{q}$ than $-\vec{1}$, i.e. $d_a(\mathcal{L}(\vec{q}),\vec{q})\leq d_a(-\vec{1},\vec{q})=\ell(\vec{q})$. Combining this, we get
\begin{equation}\label{eq:louvain_dist}
    d_a(\mathcal{L}(\vec{q}),\vec{q})\in[d_a(\mathcal{C}(\vec{q}),\vec{q}),\ell(\vec{q})].
\end{equation}

\paragraph{Projections methods.}
A mapping $\mathcal{A}$ is a \emph{projection} if $\mathcal{A}(\mathcal{A}(\vec{x}))=\mathcal{A}(\vec{x})$ holds for all $\vec{x}\in\mathbb{R}^N$.
It can easily be seen that the hypersphere projection $\mathcal{H}$ and the parallel projection $\mathcal{P}_\lambda$, for any $\lambda\in[0,\pi]$, are indeed projections.
Also, it holds that $\mathcal{C}$ is a projection onto the set of clustering vectors, as $\mathcal{C}(\vec{b}(C))=\vec{b}(C)$ holds for all clustering vectors. In Appendix~\ref{apx:proofs}, we furthermore prove that the Louvain algorithm is also a projection.
Therefore, modularity maximization using the Louvain algorithm belongs to the broader class of \emph{projection methods} that we define now:

\begin{definition}
A community detection method is a \emph{projection method} if it can be described by the following two-step approach: 1) the graph is first mapped to a query vector, and 2) the query vector is projected to the set of clustering vectors.
\end{definition}

In Section~\ref{sec:query_mappings}, we discuss various interesting options for the first step, while we use the Louvain algorithm for the second step.

\paragraph{Resolution and latitude.}
The resolution parameter in the modularity function is closely related to the latitude of the corresponding modularity vector, as described by the following lemma:
\begin{lemma}\label{lem:modularity_latitude}
    The line $(\mathcal{H}(\vec{q}_M^{\mathcal{N}}(G;\gamma)))_{\gamma\geq0}$ 
    is the (hyperspherically) straight line that starts from $\mathcal{H}(\vec{1}+\vec{e}(G))$, intersects the equator at $\gamma=1$ and ends in $\mathcal{H}(-\vec{p}^{\mathcal{N}}(G))$. The resolution parameter relates to the latitude as
    \begin{equation}\label{eq:latitude_modularity_vector}
    \ell(\vec{q}_M^{\mathcal{N}}(G;\gamma))=\arccos\left(\frac{(\gamma-1)m_G}{\sqrt{N}\cdot\sqrt{(1-\gamma)m_G+\gamma^2\left\|\vec{p}^{\mathcal{N}}(G)\right\|^2-\gamma\langle\vec{p}^{\mathcal{N}}(G),\vec{e}(G)\rangle}}\right).
\end{equation}
\end{lemma}
\begin{proof}
    For $\gamma=0$, the modularity vector is given by $\vec{q}_M^{\mathcal{N}}(G;0)=\vec{1}+\vec{e}(G)$, corresponding to the starting point of the line. The endpoint is obtained by
    \begin{align*}
    \lim_{\gamma\rightarrow\infty}\mathcal{H}(\vec{q}_M^{\mathcal{N}}(G;\gamma))=\lim_{\gamma\rightarrow\infty}\sqrt{N}\frac{\vec{1}+\vec{e}(G)-2\gamma\vec{p}^{\mathcal{N}}(G)}{\|\vec{1}+\vec{e}(G)-2\gamma\vec{p}^{\mathcal{N}}(G\|}
    =-\sqrt{N}\frac{\vec{p}^{\mathcal{N}}(G)}{\left\|\vec{p}^{\mathcal{N}}(G)\right\|}=\mathcal{H}(-\vec{p}^{\mathcal{N}}(G)).
    \end{align*}
To prove that the curve $\gamma\mapsto \mathcal{H}(\vec{q}_M^{\mathcal{N}}(G;\gamma))$ is a hyperspherical straight line, we must show that it is a segment of a great circle.
    Note that according to \eqref{eq:modularity_vector} the line $(\vec{q}_M^{\mathcal{N}}(G;\gamma))_{\gamma\geq0}$ is contained in the 2-dimensional plane that is defined by the three vectors $\vec{q}_M^{\mathcal{N}}(G;0)$, $-\vec{p}^{\mathcal{N}}(G)$ and the origin. The intersection between this plane and the hypersphere defines a great circle. Since the projection $(\mathcal{H}(\vec{q}_M^{\mathcal{N}}(G;\gamma)))_{\gamma\geq0}$ is contained in the same plane and lies on the hypersphere, it must be a subset of the same great circle. Therefore, $(\mathcal{H}(\vec{q}_M^{\mathcal{N}}(G;\gamma)))_{\gamma\geq0}$ is a segment of a great circle, so that it is a hyperspherical straight line by definition.
    Finally, \eqref{eq:latitude_modularity_vector} is obtained by substituting the modularity vector, given by \eqref{eq:modularity_vector}, into the latitude $\ell$, as given by \eqref{eq:latitude} and using $\langle\vec{p}^{\mathcal{N}}(G),\vec{1}\rangle=m_G$. 
\end{proof}
In Lemma~\ref{lem:modularity_latitude}, we restrict to $\gamma\geq0$. While modularity could also be defined for {\em negative} resolution parameter values, the optimization of modularity for $\gamma<0$ always leads to clustering all items in the same cluster (i.e., the coarse pole).
Lemma~\ref{lem:modularity_latitude} can easily be extended to take these negative resolution parameters to show that $(\mathcal{H}(\vec{q}_M^{\mathcal{N}}(G;\gamma)))_{\gamma\in\mathbb{R}}$ corresponds to the unique hyperspherical straight line from $\mathcal{H}(\vec{p}^{\mathcal{N}}(G))$ to $\mathcal{H}(-\vec{p}^{\mathcal{N}}(G))$ that passes through $\mathcal{H}(\vec{1}+\vec{e}(G))$.

Figure~\ref{fig:modularity_lines} provides a three-dimensional illustration of what these modularity lines may look like. 
For the ER null-model, we have $\mathcal{H}(\vec{p}^{\mathcal{N}}(G))=\vec{1}$, so Lemma~\ref{lem:modularity_latitude} tells us that the line $\mathcal{H}(\vec{q}_M^{\textrm{ER}}(G;\gamma))_{\gamma\geq0}$ lies on a single meridian, ranging from ${\mathcal{H}(\vec{1}+\vec{e}(G))}$ to $-\vec{1}$. The latitudes of vectors on this meridian are given by
\[
\tan\ell(\vec{q}_M^{\textrm{ER}}(G;\gamma))=\frac{\sqrt{\frac{N-m_G}{m_G}}}{\gamma-1},
\]
as is derived in Appendix~\ref{apx:proofs}. For other null models such as CM, this line does not correspond to a meridian, but passes through a range of meridians.

\begin{figure}
    \centering
    \resizebox{0.5\textwidth}{!}{\begin{tikzpicture} 


\def\R{3} 
\def\angEl{5} 
\def\angAz{-101} 
\def\xlong{48}
\def\xlat{25}
\def\bTlat{40}


\pgfmathsetmacro\H{\R*cos(\angEl)} 
\tikzset{xyplane/.style={cm={cos(\angAz),sin(\angAz)*sin(\angEl),-sin(\angAz),
                              cos(\angAz)*sin(\angEl),(0,-\H)}}}
\LongitudePlane[xzplane]{\angEl}{\angAz}
\LongitudePlane[pzplane]{\angEl}{\angAz+\xlong}
\LatitudePlane[equator]{\angEl}{0}


\fill[ball color=white] (0,0) circle (\R); 
\draw (0,0) circle (\R);


\coordinate[coordinate] (O) at (0,0);
\coordinate[mark coordinate] (N) at (0,\H);
\coordinate[mark coordinate] (S) at (0,-\H);
\path[xzplane] (\bTlat:\R) coordinate[mark coordinate] (bT);
\path[xzplane] (0:\R) coordinate[mark coordinate] (q_ER);
\pgfmathsetmacro\cosx{cos(\xlat)};
\pgfmathsetmacro\cosbT{cos(\bTlat)};
\path[xyplane] (\xlong:\R*\cosx) coordinate (zhat);
\path[xyplane] (0:\R*\cosbT) coordinate (xhat);
\path[equator] (\angAz+\xlong/2:\R) coordinate[mark coordinate] (q_CM);

\pgfmathsinandcos\sinzlat\coszlat{15}
\pgfmathsinandcos\sinxlon\cosxlon{\angAz}

\pgfmathsinandcos\sinEl\cosEl{\angEl} 
 \pgfmathsinandcos\sinAz\cosAz{\angAz} 
 \pgfmathsinandcos\sinxlat\cosxlat{\xlat}
 \pgfmathsinandcos\sinxlon\cosxlon{\xlong+\angAz}

\pgfmathsetmacro\perplength{\R/7};


\DrawLatitudeCircle[\R]{0} 
\DrawLongitudeCircle[\R]{\angAz} 




\path (S) +(0.4ex,-0.4ex) node[below] {\small $-\vec{1}$};
\path (bT) node[above left] {\small $\vec{1}+\vec{e}(G)$};
\path (q_ER) node[below left] {\small $\vec{q}_M^{\textrm{ER}}(G)$};
\path (q_CM) node[above right] {\small $\vec{q}_M^{\textrm{CM}}(G)$};

\projection\angEl{\angAz-16}{\xlong-90}\labelx\labely;

\pgfmathsetmacro\projectionlat{atan(cos(\xlong)*tan(\bTlat+90))-90-\bTlat};
\pgfmathsetmacro\heuristiclat{acos(cos(\xlong)*cos(\bTlat+90)/(1+sin(\xlong)*sin(\bTlat+90))};
\DrawArcNoMark{\angEl}{\angAz}{\bTlat}{\xlong}{-\bTlat-\bTlat}{}{}{}{};
\draw[xzplane,->,line width=0.35mm] (bT) arc (\bTlat:-\bTlat:\R);
\end{tikzpicture}
    }
    \caption{Illustration of the hyperspherical lines formed by varying the latitude of the modularity vector for a null model. The ER-modularity vectors lie on a single meridian, in contrast to the CM-modularity vectors.}
    \label{fig:modularity_lines}
\end{figure}
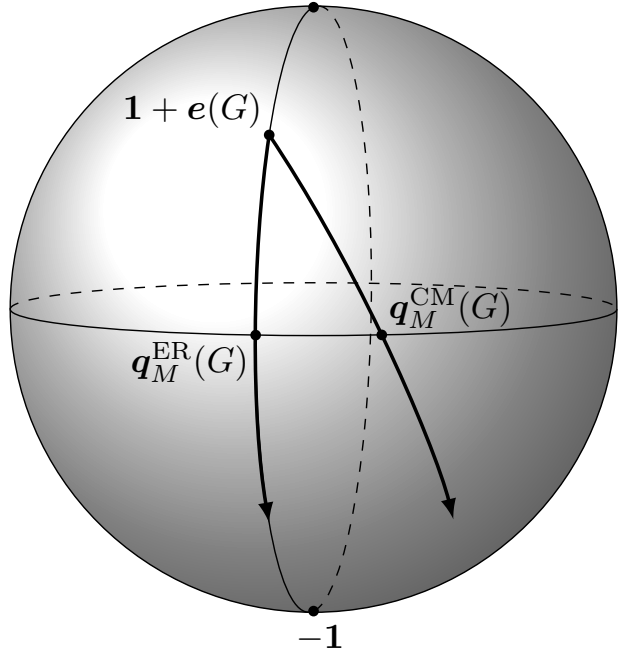

Already from the example of the ER and CM null models, we see that the latitude of the modularity vector, described by \eqref{eq:latitude_modularity_vector} in Lemma~\ref{lem:modularity_latitude}, depends on the particular null model. In other words,  two modularity vectors with the same resolution parameter but different null models have different latitudes. 
The advantage of latitude is that it puts all these modularity and clustering vectors on the same scale, allowing for a comparison of modularity-based methods of different null models.

\paragraph{The hypersphere interpretation of the resolution limit.}
A well-known limitation of modularity-maximizing methods is the so-called \emph{resolution limit}~\citep{Fortunato2007}. 
The easiest example to demonstrate this resolution limit is by a ring of cliques~\citep{Fortunato2007}: consider a graph $G_{k,s}$ that is given by a ring of $k$ cliques each consisting of $s$ vertices and let each pair of neighboring cliques be connected by a single edge. Such a graph has a natural clustering $T_{k,s}$ into its $k$ cliques. It can be shown that for any $s>2$ and sufficiently high $k$, a clustering that merges two neighboring cliques has a higher modularity value, which is clearly undesirable. While this result was initially only given for the CM null model and $\gamma=1$, it easily generalizes to other null models, resolution parameter values and graph designs~\citep{Fortunato2007,traag2011narrow}. 

The general problem is that in large sparse graphs, the expected number of edges between small communities is so negligibly small that the existence of \emph{any} edge between them is enough to result in a modularity-increase for merging these communities.
In our geometric framework, we will show that this resolution limit can be characterized as a discrepancy between the latitude of the query vector and the latitude of the ground truth clustering vector. The inverse triangle inequality allows us to bound the angular distance by
\[
d_a(\vec{b}(T),\vec{q}_M^{\mathcal{N}}(G;\gamma))\geq |d_a(-\vec{1},\vec{q}_M^{\mathcal{N}}(G;\gamma))-d_a(-\vec{1},\vec{b}(T))|=|\ell(\vec{q}_M^{\mathcal{N}}(G;\gamma))-\ell(\vec{b}(T))|.
\]
Note that for $\gamma=1$, the modularity vector has latitude $\pi/2$ while the latitude of the ground truth clustering vector  decreases roughly as $O(n^{-1/2})$, as shown in~\eqref{eq:clustering_concentration}.
Therefore, $d_a(\vec{b}(T),\vec{q}_M^{\mathcal{N}}(G;1))\approx\pi/2$ for large graphs. This means that on this hypersphere, the vectors $\vec{b}(T)$ and $\vec{q}_M^{\mathcal{N}}(G;\gamma)$ are, in a way, \emph{half a world apart}. Since modularity maximization returns a clustering $C$ with $\vec{b}(C)$ near $\vec{q}_M^{\mathcal{N}}(G;\gamma)$, this clustering $C$ will likely be far from the ground truth clustering $T$.
This characterizes the resolution limit as a discrepancy between the latitudes of the modularity vector and ground truth clustering vector.

However, using a different choice for the query vector could avoid such a resolution limit altogether: for the particular example of the ring of cliques $G_{k,s}$, the simple query mapping $\vec{q}(G)=\vec{e}(G)$ would have angular distance
\[
d_a(\vec{b}(T),\vec{e}(G_{k,s}))=\arccos\left(1-2\frac{k}{N}\right)\approx2\sqrt{k/{k\cdot s\choose 2}}=O(k^{-1/2}),
\]
for $k\rightarrow\infty$. Therefore, this query mapping overcomes the resolution limit and is thus better than the modularity query mapping for this particular graph. We emphasize that the choice of the query vector should depend on characteristics of the graph and community structure: for more realistic networks with larger and less dense communities, the query mapping $\vec{q}(G)=\vec{e}(G)$ performs poorly as its latitude is small compared to that of the ground truth clustering.
We note that $\vec{e}(G)$ does lie on the same meridian as the ER-modularity vector\footnote{More precisely, $\vec{e}(G)=\vec{q}_M^{\textrm{ER}}\left(G;\tfrac{N}{2m_G}\right)$.} while its latitude decreases as $O(n^{-1/2})$, which suggests that adjusting the latitude of a query vector is a suitable way of overcoming such resolution limits.

Based on the discussion above, we conclude that the latitude of the query vector is a more informative quantity than the resolution parameter.
Indeed, the latitude has the  advantage that it is defined for any query and clustering vector, while the resolution parameter only tells us something about the specific query vector $\vec{q}_M^{\mathcal{N}}(G;\gamma)$.
Therefore, we suggest to use the latitude of the query vector, rather than the resolution parameter, as the quantity that regulates the granularity of the candidate clustering.

\section{Beyond modularity-based methods}\label{sec:query_mappings}
As shown in Section~\ref{sec:modularity_distance}, modularity merely corresponds to a subclass of possible query mappings.
In this section, we discuss useful query mappings that do not fall under the classical formulation in terms of null models and resolution parameters.
We will argue that for this wider class of projection methods, it is unnatural to think in terms of null models and resolution parameters, but that these methods instead are better characterized in terms of the meridians and latitudes of the corresponding query vectors.

\subsection{Beyond null models}\label{sec:beyond_null}
From Lemma~\ref{lem:modularity_latitude}, we learned that for different null models, modularity corresponds to different lines on the hypersphere, each starting from $\mathcal{H}(\vec{1}+\vec{e}(G))$ and crossing the equator at different points. The line of ER modularity corresponds to the meridian of $\vec{e}(G)$, while the line of CM modularity passes through a range of meridians and intersects the ER meridian at $\mathcal{H}(\vec{1}+\vec{e}(G))$, as illustrated by Figure~\ref{fig:modularity_lines}.
Therefore, linear combinations of these two modularity vectors define a two-dimensional subspace of the hypersphere.
We observe that the best-performing query vectors among this hyperplane generally do not lie on one of those two modularity lines. As an example, we take the well-known Karate network~\citep{zachary1977information} and compute the performance (in terms of the correlation $\textrm{CC}(T,C)$, as defined in~\eqref{eq:CC}) for a range of query vectors in this two-dimensional set. The results are shown in Figure~\ref{fig:karate_heatmap}. We see that both ER and CM modularity perform reasonably well, but are outperformed by other query vectors. We further see that there is a region of query vectors above the CM-modularity line for which the true community structure is recovered. Figure~\ref{fig:football_heatmap} shows the same experiment for the Football network~\citep{girvan2002community}, where each vertex represents an American college football team, each edge represents a match played between the teams, and each community corresponds to a so-called conference. Again, the ER- and CM-modularity lines are mostly outside of the  region of best performing queries.

Note the sharp transition in Figure~\ref{fig:karate_heatmap}, where the region of `good' query vectors directly neighbors the region of poorly performing query vectors, i.e., the query vectors that lead to correlation $0$. This is easily explained by the fact that the perfectly performing query vectors yield candidates consisting of two communities. Therefore, increasing the query latitude eventually leads to these two communities merging into the clustering consisting of a single community, which has zero correlation to the true clustering. This behavior is different when the ground truth clustering has more than two communities. Indeed,  in Figure~\ref{fig:football_heatmap}, the best-performing query vectors for the Football networks do not lie right next to poorly-performing query vectors since the ground truth clustering consists of 12 communities.

\paragraph{Null models.}
Since the query vectors from Figures~\ref{fig:karate_heatmap} and~\ref{fig:football_heatmap} are linear combinations of the ER- and CM-modularity vector, one may be tempted to think that these query vectors themselves correspond to modularity vectors with some null model that is a mixture of ER and CM. However, this is generally not the case. Specifically, although for each of these query vectors $\vec{q}$ there exist $\gamma,c_1,c_2$ such that  
\[
\vec{q}=\vec{1}+\vec{e}(G)-2\gamma(c_1\vec{p}^{\textrm{ER}}(G)+c_2\vec{p}^{\textrm{CM}}(G)),
\]
it is not generally the case that $c_1$ and $c_2$ are {\em positive}. Because of this, the `expected number of edges' may be negative for certain vertex-pairs. Therefore, this linear combination does not fit the requirements corresponding to a null model. For example, each of the perfectly-performing query vectors of Figure~\ref{fig:karate_heatmap} corresponds to negative values of $c_1$. In Appendix~\ref{apx:dolphins}, we show that for another network, the best-performing query vectors correspond to negative values of $c_2$. 
For this reason, it is not natural to think of these query vectors as corresponding to some null model. Instead, they may simply be viewed as points on the hypersphere that have a good relative position to the ground truth clustering vector.

\begin{figure}
    \centering
    \includegraphics[width=0.65\textwidth]{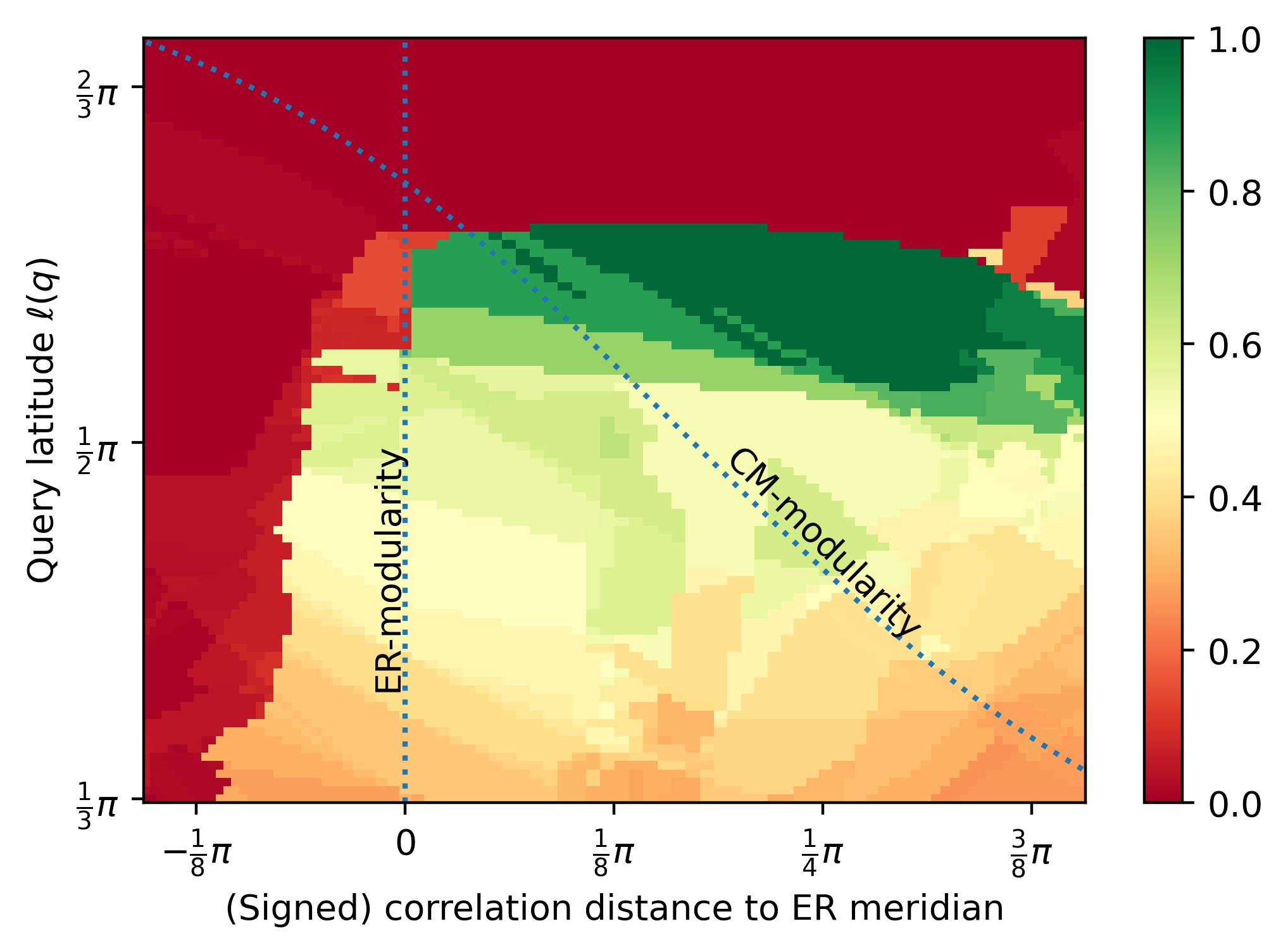}
    \caption{Heatmap of the correlation between the ground truth and the candidate clusterings for the Karate network. We take query vectors from the meridians that $\left(\vec{q}_M^{\textrm{CM}}(G;\gamma)\right)_{\gamma\in[-1.5,2]}$ runs through and vary the query latitude between $\tfrac{1}{3}\pi$ and $\tfrac{2}{3}\pi$. The horizontal coordinates are given by $\textrm{sgn}(\gamma)\cdot\dCC(\vec{q},\vec{e}(G))$, i.e. the (signed) correlation distance to the ER meridian.
    }
    \label{fig:karate_heatmap}
\end{figure}

\begin{figure}
    \centering
    \includegraphics[width=0.65\textwidth]{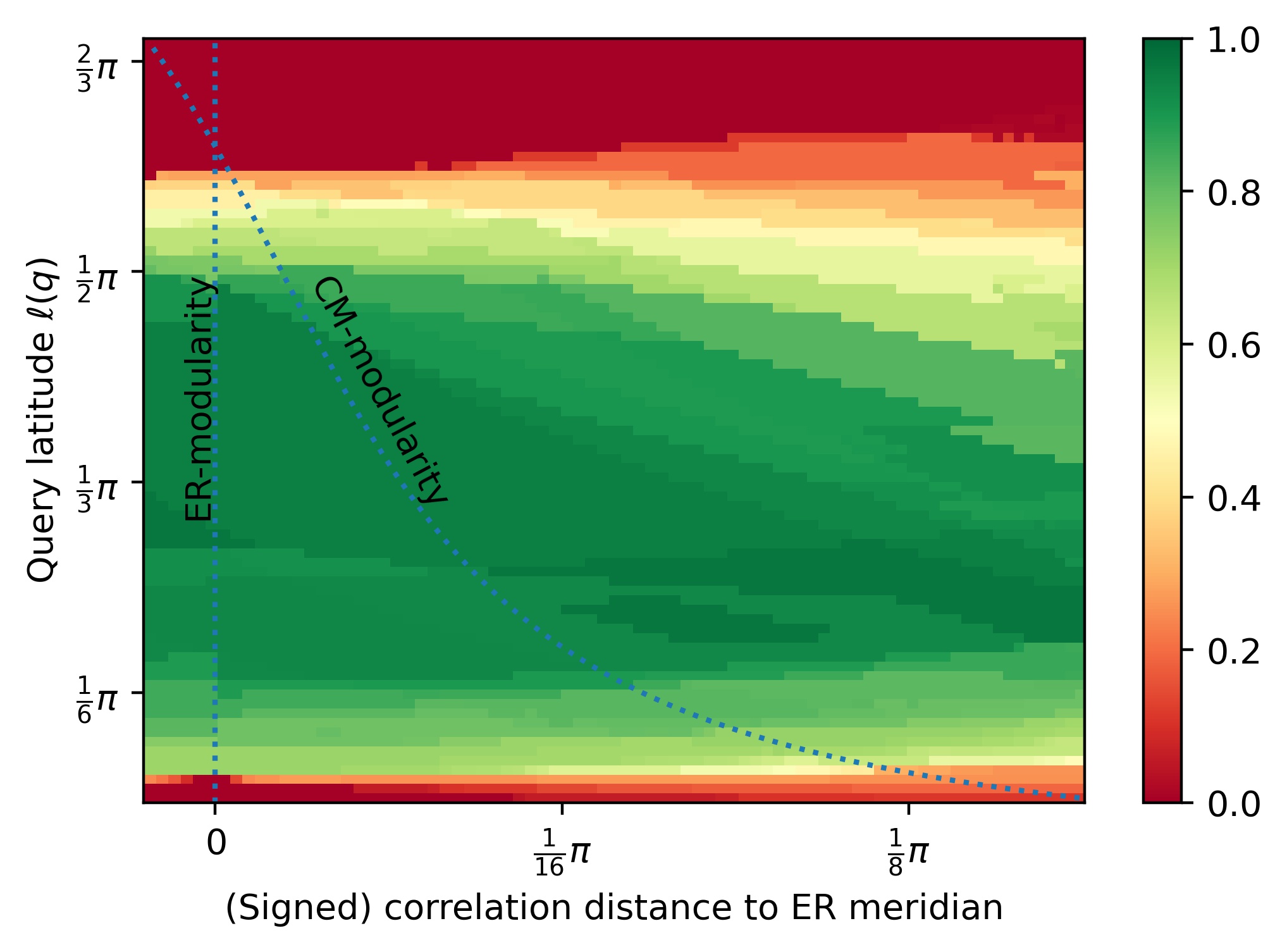}
    \caption{Heatmap of the correlation between the ground truth and the candidate clusterings for the Football network. We take query vectors from the meridians that $\left(\vec{q}_M^{\textrm{CM}}(G;\gamma)\right)_{\gamma\in[-1,14]}$ runs through. The horizontal coordinates are given by $\textrm{sgn}(\gamma)\cdot\dCC(\vec{q},\vec{e}(G))$, i.e. the (signed) correlation distance to the ER meridian. 
    }
    \label{fig:football_heatmap}
\end{figure}

\subsection{Query mappings based on common neighbors}\label{sec:wedges}
Similarly to how we expressed the connectivity of a graph $G$ by the edge-connectivity query vector $\vec{e}(G)$ in Section~\ref{sec:modularity_distance}, we can express the number of common neighbors between each pair of vertices in terms of a vector as well.
We denote this vector by $\vec{w}(G)$ and refer to it as the \emph{wedge vector}, because the entry $\vec{w}(G)_{ij}$ corresponds to the number of distinct \emph{wedges} (paths of length 2) that have the vertices $i$ and $j$ as endpoints. Intuitively, $\vec{w}(G)$ is closely related to the community structure of $G$ because disconnected vertices $i$ and $j$ of the same community may have many common neighbors, which makes $\vec{w}(G)_{ij}$ high as well, and thus $\vec{b}(C)$ tends to be closer to $\vec{w}(G)$ when clustering $C$ puts $i$ and $j$ in the same community.  

To provide further motivation for the claim that this wedge vector is of interest, we relate it to the \emph{global clustering coefficient} of a graph $G$, which is defined as thrice the number of triangles (i.e., complete subgraphs of size $3$) divided by the number of wedges (i.e., pairs of adjacent vertices). Note that each triangle consists of three wedges, so that this quantity equals $1$ whenever the graph consists of disconnected cliques, and $0$ when the graph is a tree (or all cycles have length at least $4$). Thus, this clustering coefficient quantifies how much the graph \emph{resembles a clustering}. We prove the following result:
\begin{lemma}
\label{lem:GlobalClusteringWedges}
    The global clustering coefficient of a graph $G$ is given by
    \[
        {\rm GlobalClustering}(G)
        =\frac{1}{2}\left(1-\frac{\cos d_a(\vec{e}(G),\vec{w}(G))}{\cos\ell(\vec{w}(G))}\right).
    \]
\end{lemma}
\begin{proof}
    The global clustering coefficient is given by the fraction of wedges that are closed. Since each triangle consists of three closed wedges, this coefficient is given by thrice the number of triangles divided by the number of wedges. 
    If there is an edge between $i$ and $j$, then this edge $ij$ is part of exactly $\vec{w}(G)_{ij}$ triangles. Summing $\vec{w}(G)_{ij}$ over all edges thus gives thrice the number of triangles, as each triangle contains three edges. In vector notation, this is given by $\langle\tfrac{1}{2}(\vec{e}(G)+\vec{1}),\vec{w}(G)\rangle$, where $(\vec{e}(G)+\vec{1})/2$ is the $\{0,1\}$-binary vector edge-connectivity vector. 
    The total number of wedges is given by $\langle\vec{1},\vec{w}(G)\rangle$. Combined, this allows us to write the clustering coefficient as
    \[
        \textrm{GlobalClustering}(G)=\frac{\langle\tfrac{1}{2}(\vec{e}(G)+\vec{1}),\vec{w}(G)\rangle}{\langle\vec{1},\vec{w}(G)\rangle}=\frac{1}{2}\left(1+\frac{\langle\vec{e}(G),\vec{w}(G)\rangle}{\langle\vec{1},\vec{w}(G)\rangle}\right).
    \]
    Finally, from the definitions of the angular distance and latitude, as given by~\eqref{eq:angular_distance} and~\eqref{eq:latitude}, it follows that
    \[
    \frac{\langle\vec{e}(G),\vec{w}(G)\rangle}{\langle\vec{1},\vec{w}(G)\rangle}=-\frac{\cos d_a(\vec{e}(G),\vec{w}(G))}{\cos\ell(\vec{w}(G))},
    \]
    which completes the proof.
\end{proof}

Since $\vec{w}(G)$ is non-negative, by \eqref{eq:latitude}, it always has a latitude at least $\pi/2$, so we always have $\cos\ell(\vec{w}(G))\leq0$. Then  Lemma~\ref{lem:GlobalClusteringWedges} tells us that the global clustering coefficient of a graph $G$ is high whenever the angular distance between $\vec{w}(G)$ and $\vec{e}(G)$ is low. This relation to the global clustering coefficient shows that the wedge vector indeed contains relevant information about our graph.
However, if we would directly use the vector $\vec{w}(G)$ as a query vector, we would run into a problem because $\vec{w}(G)$ has a latitude at least $\pi/2$, and usually much higher than $\pi/2$, so the resulting Louvain candidate $\mathcal{L}(\vec{w}(G))$ groups all vertices into the same community.

The most straightforward option is to project $\vec{w}(G)$ to the equator, resulting in the query mapping $\vec{q}_w(G)=\mathcal{P}_{\pi/2}(\vec{w}(G))$. The motivation for this choice is that, by
Lemma~\ref{lem:modularity_latitude}, the modularity vector has latitude $\pi/2$ for the default resolution parameter value $\gamma=1$. Thus, the projection to the equator is comparable to the most common form of modularity.  Being on the equator implies that $\langle\vec{q}_w(G),\vec{1}\rangle=0$. Then, the positive entries of $\vec{q}_w(G)$ correspond to vertex-pairs that have more common neighbors than the graph average (over the vertex-pairs), while negative entries correspond to vertex-pairs that have less common neighbors than average. An approximate nearest neighbor of $\vec{q}_w(G)$ then corresponds to a clustering where many of the intra-cluster pairs have more common neighbors than average, while many of the inter-cluster pairs have fewer common neighbors than average. This wedge vector can also be projected to other latitudes to obtain a clustering of the desired granularity.

\begin{figure}
    \centering
    \includesvg[width=0.7\textwidth]{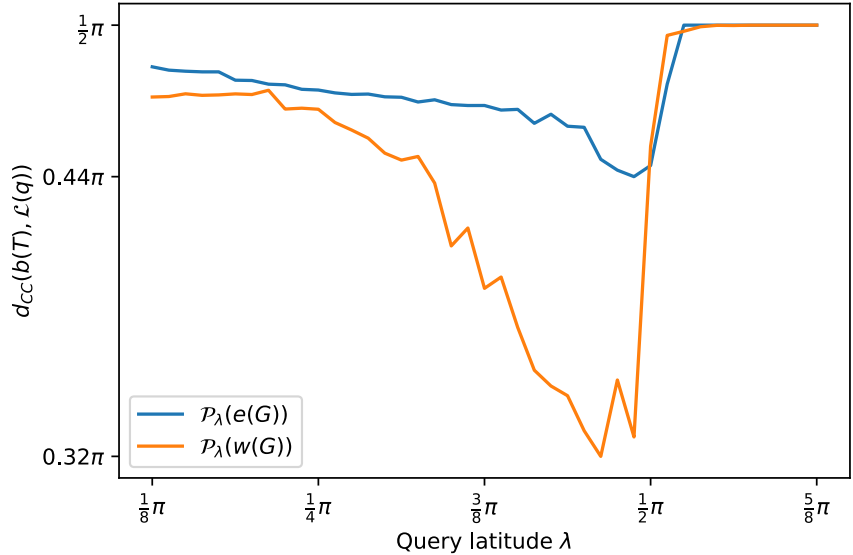}
    \caption{Comparing the performances of query vectors based on the wedge vector $\vec{w}(G)$ and edge vector $\vec{e}(G)$ for a PPM consisting of $10$ communities, each of size $100$, with expected inter- and intra-community degree 5.}
    \label{fig:wedges_experiment}
\end{figure}

\paragraph{Detecting large communities.}
We observe that the wedge vector is especially useful as query vector when the community sizes are significantly larger than the average degrees, or equivalently, when $\ell(\vec{b}(T))$ is significantly larger than $\ell(\vec{e}(G))$. In Figure~\ref{fig:wedges_experiment} we compare wedge-based query vectors to edge-based query vectors on a Planted Partition Model\footnote{That is, a Stochastic Block Model where the block matrix has diagonal entries $p_{\text{in}}$ and off-diagonal entries $p_{\text{out}}$.} (PPM) with 10 communities each of size $100$ and average degree $10$, which gives $\ell(\vec{e}(G))=0.064\pi$ and $\ell(\vec{b}(T))=0.204\pi$. We see that, for the right query latitude, the wedge-based query vectors result in candidates with correlation distance $0.32\pi$ to the ground truth clustering, while the edge-based query vectors only result in candidates with correlation distance $0.44\pi$, which translate to correlation coefficients of $0.54$ and $0.20$ respectively. 
This significant difference in performance may be explained by the fact that the edge vector only connects each vertex to 5 of its 99 community members on average, while the wedge vector connects each vector to approximately 25 community-members.

\paragraph{Computational cost.}
A practical downside of projection methods based on wedge vectors compared to edge vectors is the higher computational cost. Louvain is known to run in roughly log-linear time in terms of the number of edges. When a network is sparse, meaning that the number of edges is of the same order as the number of vertices, this results in a computational time that is log-linear  in the number of vertices.
Similarly, when applying Louvain to the wedge vector, the computational time is roughly log-linear in the number of vertex-pairs that have wedges between them. This can be upper-bounded by the total number of wedges, which is given by
\[
\langle\vec{w}(G),\vec{1}\rangle=\sum_{i\in[n]}{d_i^{(G)}\choose2}.
\]
If the degree distribution of the network has a finite second moment, then the expected value of this sum is $\bigO(n)$, so that Louvain will run in $\bigO(n\log n)$. However, many real-world networks are known to be scale-free~\citep{barabasi2013network,fortunato2010community,stegehuis2016power,voitalov2019scale}, meaning that their degree distribution has a finite mean, but infinite second moment. 
In such cases, the total number of wedges is $\Omega(n)$, leading to a higher computational costs for projection methods based on wedge vectors.
For this reason, using projection methods based on the wedge vector may be computationally infeasible in large scale-free networks.

\section{Empirical analysis of projection methods}\label{sec:experiments}
In this section, we perform more  experiments to demonstrate how the geometric results reported in this paper help to interpret and understand outcomes of different projection methods. In Section~\ref{sec:observations}, we illustrate some phenomena that we empirically  observe for many networks and projection methods, while in Section~\ref{sec:realworld}, we compare the performances of projection methods with different query mappings on several real-world networks.

The code that was used to perform the experiments and generate the figures of this paper is available on GitHub\footnote{\url{https://github.com/MartijnGosgens/hyperspherical_community_detection}}. Amongst others, this repository contains a Python implementation of our modification of the Louvain algorithm. This implementation is able to compute the Louvain projection for a range of query vectors, including all query vectors that are linear combinations of $\vec{1}$, $\vec{e}(G)$, $\vec{w}(G)$, $\vec{q}_M^{\textrm{CM}}(G;\gamma)$, and $\vec{q}_M^{\textrm{ER}}(G;\gamma)$.

\subsection{The Louvain candidate}\label{sec:observations}
In this section, we describe how several geometric quantities are affected when the query latitude is varied. The relevant quantities are summarized in Table~\ref{tab:quantities}.
For an illustrational experiment, we use query vectors on the ER meridian to detect communities in a Planted Partition Model (PPM). The ground truth clustering consists of $20$ communities, each of size $20$, while the mean intra- and inter-community degrees are $6$ and $4$, respectively.
The four plots in  Figure~\ref{fig:ppm_experiment} illustrate our empirical observations, which we now explain in detail.

\begin{table}[]
    \centering
    \begin{footnotesize}
    \begin{tabular}{c|L{11cm}}
\toprule
Quantity &  Description \\
\midrule
$\ell(\vec{b}(C))$ & Latitude of candidate clustering, measure of granularity.\\
$\ell(\vec{q})$ & Latitude of query vector. Related to resolution parameter for modularity vectors.\\
$d_a(\vec{q},\vec{b}(C))$ & Angular distance between query vector and candidate clustering, the quantity that is minimized by Louvain.\\
$d_a(\vec{q},\vec{b}(T))$ & Angular distance between the query vector and the ground truth clustering.\\
$\dCC(\vec{q},\vec{b}(C))$ & Correlation distance between the query vector and the candidate clustering. \\
$\dCC(\vec{q},\vec{b}(T))$ & Correlation distance between the query vector and the ground truth clustering, measure of how informative $\vec{q}$ is for detecting $T$.\\
$\dCC(\vec{b}(T),\vec{b}(C))$ & Correlation distance between the ground truth and the candidate clusterings, the clustering performance measure used in this article.\\
\bottomrule
\end{tabular}
    \end{footnotesize}
    \caption{Description of quantities in terms of the query vector $\vec{q}$ and the corresponding Louvain candidate clustering $\vec{b}(C)=\mathcal{L}(\vec{q})$.}
    \label{tab:quantities}
\end{table}

\begin{figure}
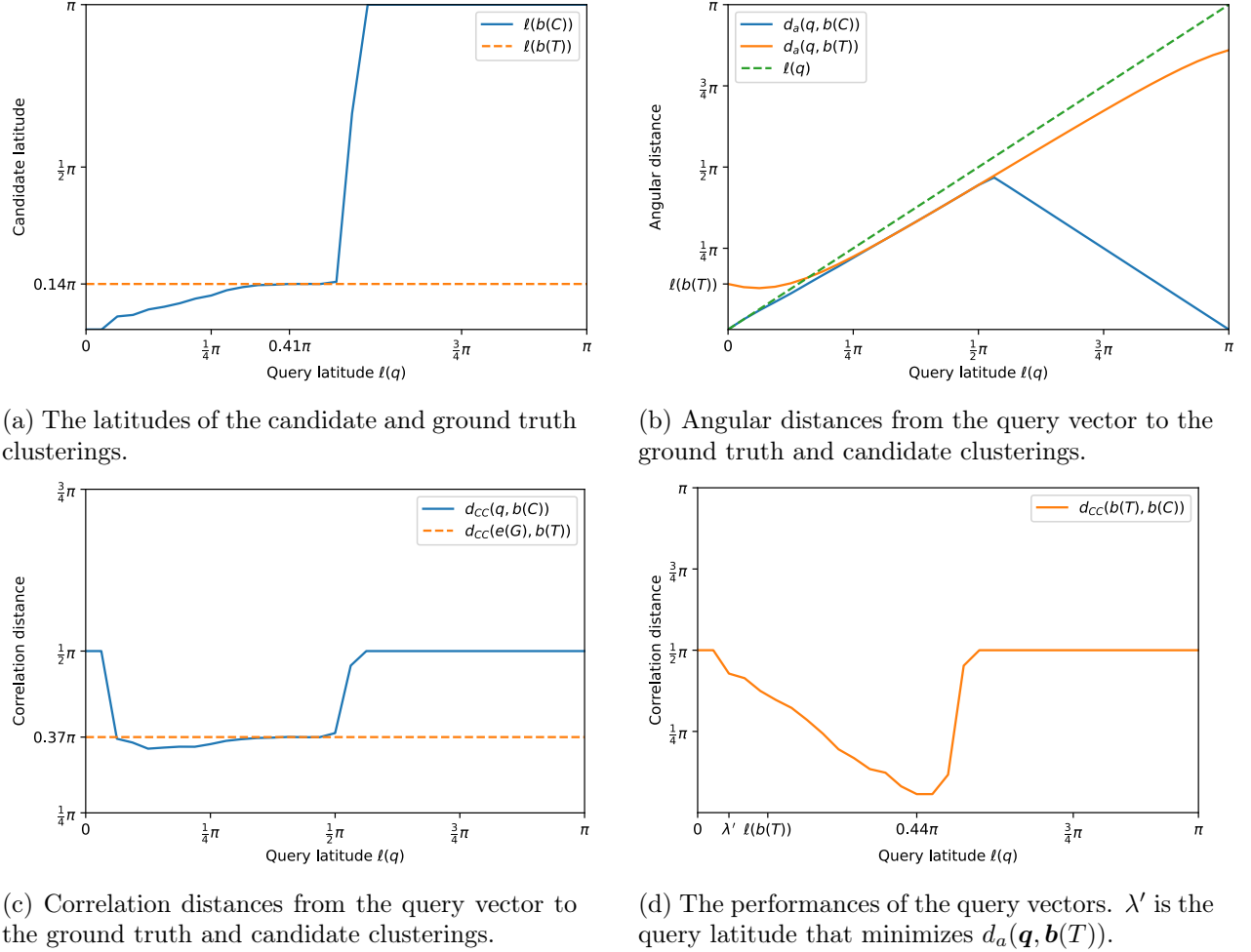

     \centering
     \begin{subfigure}[b]{0.475\textwidth}
         \centering
         \includesvg[height=5.5cm]{svgs/ppm_experiment_cluster_lats}
         \caption[c]{The latitudes of the candidate and ground truth clusterings.}
         \label{fig:ppm_experiment_cluster_lats}
     \end{subfigure}
     \hfill
     \begin{subfigure}[b]{0.475\textwidth}
         \centering
         \includesvg[height=5.5cm]{svgs/ppm_experiment_angular_dists}
         \caption[c]{Angular distances from the query vector to the ground truth and candidate clusterings.}
         \label{fig:ppm_experiment_angular_dists}
     \end{subfigure}
    \vskip\parskip
    \begin{subfigure}[b]{0.475\textwidth}
         \centering
         \includesvg[height=5.5cm]{svgs/ppm_experiment_correlation_dists}
         \caption[c]{Correlation distances from the query vector to the ground truth and candidate clusterings.}
         \label{fig:ppm_experiment_correlation_dists}
     \end{subfigure}
     \hfill
     \begin{subfigure}[b]{0.475\textwidth}
         \centering
         \includesvg[height=5.5cm]{svgs/ppm_experiment_validation}
         \caption[c]{The performances of the query vectors. $\lambda'$ is the query latitude that minimizes $d_a(\vec{q},\vec{b}(T))$.}
         \label{fig:ppm_experiment_correlation_triangle}
     \end{subfigure}
    \vskip\parskip
        \caption{Plot of various quantities for a PPM consisting of $20$ communities each of size $20$, with mean intra-community degree $6$ and inter-community degree $4$. The query vectors are given by $\vec{q}=\mathcal{P}_\lambda(\vec{e}(G))$, where $\lambda$ is varied between $0$ and $\pi$, so that these query vectors lie on the ER meridian (see~\eqref{eq:parallel_projection}). The candidate clusterings are obtained using the Louvain algorithm, i.e., $\vec{b}(C)=\mathcal{L}(\vec{q})$.}
        \label{fig:ppm_experiment}
\end{figure}

\paragraph{Candidate latitude.}
In Figure~\ref{fig:ppm_experiment_cluster_lats}, we show how the latitude of the candidate clustering changes with the latitude of the query vector. We see that the candidate latitude equals $0$ when the query latitude is $0$ (since Louvain is a projection, and latitude 0 corresponds to the fine pole where each vertex forms its own cluster). Then the latitude of the candidate gradually increases until it roughly reaches the ground truth latitude (the dashed orange line), plateaus for a bit and then sharply shoots up to $\pi$.
This is exactly what we could expect: in the part before the plateau, Louvain clusters together only a few subgraphs with a high edge-density, while after the plateau, Louvain starts to merge communities together.

\paragraph{Angular distances w.r.t.\ query vector.}
Figure~\ref{fig:ppm_experiment_angular_dists} compares the angular distances between query vector and ground truth clustering, and between query vector and candidate clustering. When the query latitude equals $0$, i.e., for $\vec{q}=-\vec{1}$, we have $d_a(\vec{q},\vec{b}(T))=\ell(\vec{b}(T))$ by definition of the latitude. The angular distance to the ground truth then slightly decreases before it steadily increases to $\pi-\ell(\vec{b}(T))$ at query latitude $\pi$. The angular distance to the candidate clustering, on the other hand, starts at $0$, exactly as in Figure~\ref{fig:ppm_experiment_cluster_lats}, and increases almost linearly to roughly $\tfrac{1}{2}\pi$, before it decreases linearly to $0$. The linear decrease corresponds to the segment where $\ell(\vec{b}(C))=\pi$, so that $d_a(\vec{q},\vec{b}(C))=\pi-\ell(\vec{q})$.

Additionally, Figure~\ref{fig:ppm_experiment_angular_dists} shows that $d_a(\vec{q},\vec{b}(C))\leq\ell(\vec{q})$ holds for all query latitudes, as claimed in~\eqref{eq:louvain_dist}. We also see that the two angular distances (the blue and the orange lines) are approximately equal to each other on the interval roughly between $\tfrac{1}{4}\pi$ and $\tfrac{1}{2}\pi$.

Furthermore, we observe that $d_a(\vec{q},\vec{b}(C))\leq d_a(\vec{q},\vec{b}(T))$ holds in all cases. That is, the candidate clustering is always at least as close to the query vector as the ground truth clustering. By the equivalence of Theorem~\ref{thm:modularity_distance}, this means that the modularity of the candidate clustering is at least as high as that of the ground-truth clustering. This behavior has previously been observed in various real-world and synthetic networks~\citep{prokhorenkova2019community}, and may be explained by the fact that the modularity landscape is \emph{glassy}~\citep{good2010performance}. That is, there are many clusterings with locally optimal modularity values very close to the global optimum. The greedy optimization that Louvain utilizes seems to be quite successful in reaching one of those local optima, while the ground truth clustering is not guaranteed to correspond to a local modularity optimum at all. This phenomenon seems to also hold when minimizing the angular distance to other query vectors. 

\paragraph{Correlation distances w.r.t.\ query vector.}
When comparing the correlation distances from the query vector to the ground truth and candidate clusterings in Figure~\ref{fig:ppm_experiment_correlation_dists}, we observe that $\dCC(\vec{q},\vec{b}(C))\approx \dCC(\vec{q},\vec{b}(T))$ for a large range of query latitudes. This is interesting since the candidate clusterings in this range do have significantly different latitudes, as shown in Figure~\ref{fig:ppm_experiment_cluster_lats}. After that, the correlation distance jumps to $\tfrac{1}{2}\pi$ and stays constant, since the candidate clustering then corresponds to the coarse pole, so that the correlation distance is $\tfrac{1}{2}\pi$ by definition.

\paragraph{Correlation distance between clusterings.}
Figure~\ref{fig:ppm_experiment_correlation_triangle} shows the correlation distance between the ground truth and candidate clusterings. We use this measure to quantify the performance of the detection algorithm.
We see that when the query latitude is zero, this quantity equals $\tfrac{1}{2}\pi$, indicating that the candidate clustering is uncorrelated to the ground truth. For higher query latitudes, this correlation distance decreases to a minimum at query latitude $0.44\pi$, after which it increases back to $\tfrac{1}{2}\pi$ again. The best performance is thus obtained around a query latitude of $0.44\pi$ for this network.

If, in advance, one would guess what query latitudes leads to the best performance, then the following two obvious guesses may come to mind: On the one hand, one may think that setting the query latitude equal to the ground truth latitude $\ell(\vec{b}(T))$ would result in a candidate of a similar latitude. 
On the other hand, the fact that Louvain finds a clustering latitude close to the query vector suggests that choosing the query latitude $\lambda'$ that minimizes the angular distance between query vector and ground truth vector would work well. It turns out that both of these options perform poorly in general, as shown in Figure~\ref{fig:ppm_experiment_correlation_triangle}. 
Instead, the best performance seems to be achieved around the query latitude for which the candidate latitude intersects the ground truth latitude in Figure~\ref{fig:ppm_experiment_cluster_lats}. In Figure~\ref{fig:ppm_experiment_cluster_lats}, we also saw that the candidate latitude is generally significantly smaller than the query latitude, so that one indeed needs query latitudes much larger than $\ell(\vec{b}(T))$ in order to obtain a candidate clustering with $\ell(\vec{b}(C))=\ell(\vec{b}(T))$. This tells us that both these initial guesses are wrong, and that instead we need to set the query latitude to some value larger than $\ell(\vec{b}(T))$ for the best performance. 

We generally observe that the difference between the best-performing query latitude and the ground truth latitude is large whenever the correlation distance $\dCC(\vec{q},\vec{b}(T)))$ between the query vector and the ground truth clustering is large. Indeed, when $\dCC(\vec{q},\vec{b}(T))=0$ and $\ell(\vec{q})=\ell(\vec{b}(T))$, it also holds that $\vec{q}=\vec{b}(T)$ so that $\mathcal{L}(\vec{q})=\vec{b}(T)$ follows from the fact that Louvain is a projection.
Of course, in practical applications, $\dCC(\vec{q},\vec{b}(T))$ and $\ell(\vec{b}(T))$ are unknown, making it difficult to know in advance which query latitude performs best.
Finding this optimal query latitude in practical settings is beyond the scope of this paper.

\subsection{Experiments on real-world networks}\label{sec:realworld}
In this section, we compare different projection methods on real-world networks. We consider both projection methods that fall inside the class of modularity-based methods and methods that fall outside this class.
We consider four sets of query mappings, each corresponding to a straight line on the hypersphere that is parametrized by the query latitude.
More specifically, we consider ER and CM modularity for various resolution parameters, the wedge-based query mappings from Section~\ref{sec:wedges} for various latitudes, and the meridian corresponding to the CM-modularity vector for resolution parameter $\gamma=1$, which we refer to as the \emph{CM meridian}.
Note that of these four lines on the hypersphere, CM modularity is the only one that does not correspond to a meridian.

For these experiments, we consider 6 real-world networks from~\cite{prokhorenkova2019community} that have known ground truth communities.
These networks are: 
1)~Zachary's well-known \emph{karate club} network~\citep{zachary1977information}, 
2)~Lusseau's network of bottlenose \emph{dolphins}~\citep{lusseau2004identifying}, 
3)~a network of \emph{political books} grouped by political affiliation~\citep{newman2006modularity}, 
4)~a network of college \emph{football} teams grouped by `conference'~\citep{girvan2002community}, 
5)~the \emph{EU-core} network of European researchers linked by email traffic and grouped by department, and
6)~a network of \emph{political blogs} concerning the US presidential election of 2004, grouped by party affiliation~\citep{adamic2005political}.
An overview of these networks is given in Table~\ref{tab:datasets}.

The repository of \cite{prokhorenkova2019community} also contains two other networks: a network of Internet systems that are grouped by a geometric clustering algorithm, and a citation network that is grouped by a text-clustering algorithm. We chose to not include these two networks in the experiments since their corresponding `ground truth' clusterings were obtained by different clustering algorithms, so that they may best be considered candidate clusterings rather than ground truth. Furthermore, these networks are significantly larger than the ones considered here ($n>20000$), which results in long running times for our modification\footnote{We implemented our algorithm in Python, which is significantly slower than C++ for such tasks.} of Louvain, especially when using the wedge-based query vector, as explained in Section~\ref{sec:wedges}.

\paragraph{The optimal query latitude.}
The results of the experiments can be found in Figure~\ref{fig:benchmarks}. We use the correlation distance between the ground truth clustering and the candidate clustering to measure the performance of the community detection method, where lower values indicate better performance.
By comparing the locations of the minima for different type of query vectors (e.g., the wedges and ER meridian in Figure~\ref{fig:football_benchmark}), we see that the optimal query latitude is dependent on the network and query type. Furthermore, while the best-performing query vectors are often located roughly around the equator, there also are many networks where the best latitude is quite far from the equator (e.g., Figures~\ref{fig:karate_benchmark} and~\ref{fig:football_benchmark}).
In Section~\ref{sec:observations}, we have discussed that the optimal query latitude is generally larger than the ground truth latitude $\ell(\vec{b}(T))$. In Figure~\ref{fig:benchmarks}, we have marked the ground truth latitudes on the horizontal axes. It can be seen that the minima are indeed located at larger query latitudes, with the exceptions of  the Football and Political blogs networks, where some minima seem to coincide with the ground truth latitude.
Note that the exact location of these minima can only be determined when the ground truth is known, which is generally not the case in practice.
Finding the best-performing query latitude without knowledge of the ground truth is beyond the scope of this article.
In the remainder of this section, we compare the four methods by comparing their minima.

\paragraph{CM meridian versus modularity.}
Recall that CM modularity corresponds to a standard community detection method, while CM meridian is one of the possible alternative community detection methods that emerge from our generalization of modularity-based methods.
Interestingly, CM meridian perfectly recovers the ground truth for Karate. In terms of Figure~\ref{fig:karate_heatmap}, this means that this meridian intersects the region of perfectly-performing query vectors. For the other networks, CM meridian seems to perform on par with CM modularity. The only exception is the EU-core, where CM modularity clearly outperforms CM meridian.

\paragraph{Correlation distance between the ground truth and query vectors.}
The correlation distance between the query vector and ground truth clustering vector seems to carry some information about which null model performs best: for all networks where CM modularity or meridian outperforms ER modularity, we see that  $\dCC(\vec{q}_M^{\textrm{CM}}(G),\vec{b}(T))<\dCC(\vec{q}_M^{\textrm{ER}}(G),\vec{b}(T))$ holds in Table~\ref{tab:datasets}. This suggests using the correlation distance between the modularity vector and the ground truth as criteria to choose a suitable null model.
However, a smaller correlation distance between the query vector and ground truth does not guarantee better performance. For example, $\vec{w}(G)$ is closer to the ground truth for all networks except EU-core (see Table~\ref{tab:datasets}), while the Football network is the only network for which the wedge vector actually outperforms the other vectors.

\paragraph{Performance of the wedge vector.}
As discussed in Section~\ref{sec:wedges}, we expected the wedge vector to perform well whenever the latitude of the ground truth vector is significantly larger than the latitude of the edge vector. We do not observe this to hold in the considered real-world networks: In Table~\ref{tab:datasets}, we see that for the Political blogs network, we have $\ell(\vec{b}(T))=0.500\pi>0.096\pi=\ell(\vec{e}(G))$ so that we would expect the wedge vector to perform well, while it performs poorly compared to the other methods. Moreover, the only network for which the wedge vector visibly outperforms the other methods is the Football network for which $\ell(\vec{e}(G))>\ell(\vec{b}(T))$. We expect that this difference in performance between the PPM of Section~\ref{sec:wedges} and the real-world networks can be explained by degree-inhomogeneity: a pair of high-degree vertices are likely to have many common-neighbors, regardless of whether they belong to the same community. This is similar to how ER modularity is known to cluster together high-degree vertices. This suggests `correcting' for this degree inhomogeneity in a similar way that CM modularity does: by subtracting a multiple of $d_i^{(G)}d_j^{(G)}$ from every entry.

In summary, we see that the developed geometry helps us to interpret results of community detection methods. In addition, we see that the class of projection methods contains many new community detection methods that may outperform existing modularity-based methods on real-world networks. Furthermore, we see that the correlation distance between the query vector and ground truth clustering can to some extend help in predicting which projection method will perform best.

\begin{table}[]
    \centering
    \begin{scriptsize}
    \begin{tabular}{l|rrrccccc}
\toprule
Dataset &    $n$ &  $m_G$ &  $|T|$ & $\ell(\vec{b}(T))$ & $\ell(\vec{e}(G))$ & $\dCC(\vec{q}_M^{\textrm{ER}}(G),\vec{b}(T))$ & $\dCC(\vec{q}_M^{\textrm{CM}}(G),\vec{b}(T))$ & $\dCC(\vec{w}(G),\vec{b}(T))$\\
\midrule
Karate club   &     34 &     78 &      2 &         $0.491\pi$ & $0.243\pi$&                      $0.400\pi$ &              $0.388\pi$    &      $0.342\pi$ \\
Dolphins &     62 &    159 &      2 &         $0.536\pi$& $0.187\pi$ &                      $0.420\pi$ &                 $0.422\pi$     &  $0.364\pi$ \\
Political books &    105 &    441 &      3 &         $0.433\pi$& $0.183\pi$ &                      $0.413\pi$ &               $0.414\pi$     &    $0.344\pi$ \\
Football &    115 &    613 &     12 &         $0.181\pi$& $0.198\pi$ &                      $0.248\pi$ &                 $0.248\pi$      &  $0.186\pi$\\
EU-core  &   1005 &  16706 &     42 &         $0.139\pi$& $0.116\pi$ &                      $0.411\pi$ &               $0.403\pi$      &    $0.444\pi$\\
Political blogs &   1224 &  16718 &      2 &         $0.500\pi$& $0.096\pi$ &                      $0.461\pi$ &          $0.458\pi$        &      $0.424\pi$ \\
\bottomrule
\end{tabular}
    \end{scriptsize}
    \caption{Overview of the considered real-world networks. For each network, we show the number of vertices $n$, the number of edges $m_G$ and the number of ground truth communities $|T|$. We also show the following angles: $\ell(\vec{b}(T))$, a measure of the granularity of the ground truth clustering; $\ell(\vec{e}(G))$, a measure of the edge-density of the network; and $\dCC(\vec{q},\vec{b}(T))$, the correlation distance between the ground truth clustering and the query vector, for $\vec{q}\in\{\vec{q}_M^{\textrm{ER}}(G),\vec{q}_M^{\textrm{CM}}(G),\vec{w}(G)\}$, where the modularity vectors use the default resolution parameter value $\gamma=1$.}
    \label{tab:datasets}
\end{table}

\begin{figure}
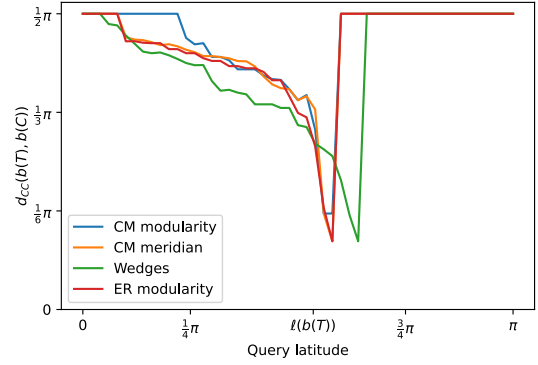
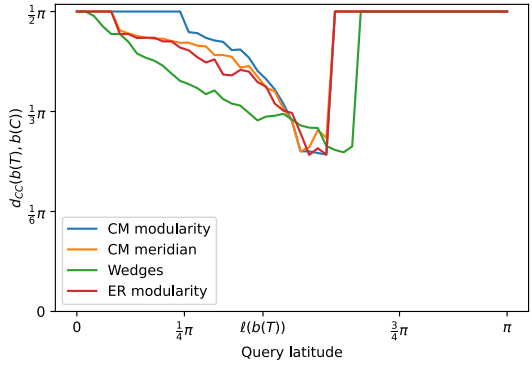
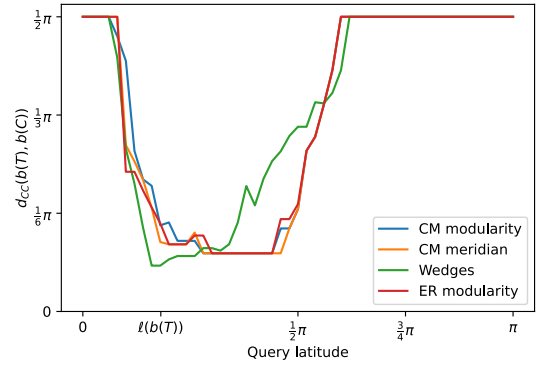
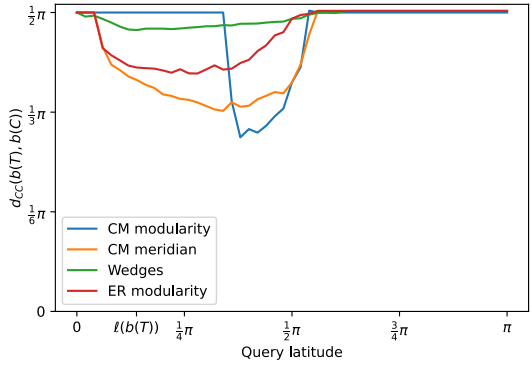
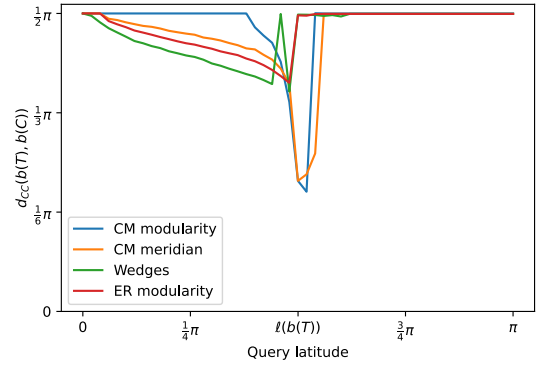

     \centering
     \begin{subfigure}[b]{0.475\textwidth}
         \centering
         \includesvg[height=5cm,inkscapename=a\thefigure]{svgs/karate_benchmark}
         \caption[c]{Performance for Karate}
         \label{fig:karate_benchmark}
     \end{subfigure}
     \hfill
     \begin{subfigure}[b]{0.475\textwidth}
         \centering
         \includesvg[height=5cm,inkscapename=b\thefigure]{svgs/dolphins_benchmark}
         \caption[c]{Performance for Dolphins}
         \label{fig:dolphins_benchmark}
     \end{subfigure}
    \vskip\baselineskip
    \begin{subfigure}[b]{0.475\textwidth}
         \centering
         \includesvg[height=5cm,inkscapename=c\thefigure]{svgs/polbooks_benchmark}
         \caption[c]{Performance for Political books}
         \label{fig:polbooks_benchmark}
     \end{subfigure}
     \hfill
     \begin{subfigure}[b]{0.475\textwidth}
         \centering
         \includesvg[height=5cm,inkscapename=d\thefigure]{svgs/football_benchmark}
         \caption[c]{Performance for Football}
         \label{fig:football_benchmark}
     \end{subfigure}
    \vskip\baselineskip
    \begin{subfigure}[b]{0.475\textwidth}
         \centering
         \includesvg[height=5cm,inkscapename=e\thefigure]{svgs/eu-core_benchmark}
         \caption[c]{Performance for EU-core}
         \label{fig:eu-core_benchmark}
     \end{subfigure}
     \hfill
     \begin{subfigure}[b]{0.475\textwidth}
         \centering
         \includesvg[height=5cm,inkscapename=f\thefigure]{svgs/polblogs_benchmark}
         \caption[c]{Performance for Political blogs}
         \label{fig:polblogs_benchmark}
     \end{subfigure}
    \vskip\baselineskip
        \caption{Results of the various projection methods on the real-world networks summarized in Table~\ref{tab:datasets}. The ground truth latitude is marked on the horizontal axis.}
        \label{fig:benchmarks}
\end{figure}

\section{Discussion}\label{sec:discussion}
In this work, we described a hyperspherical geometry on clusterings and showed how this geometry is related to validation measures such as the correlation distance. We then extended this geometry to include vectors that do not necessarily correspond to clusterings and proved that modularity maximization is equivalent to minimizing the distance to some modularity vector over the set of clustering vectors.
In Section~\ref{sec:consequences}, we discussed how this newfound geometry sheds new light on modularity-based community detection methods: it allows us to view the popular Louvain algorithm as a method that projects a \emph{query vector} onto the set of clustering vectors. In addition, this led to a geometric interpretation of the resolution limit. This geometry also suggests a generalization of modularity-based community detection methods, leading to the class of \emph{projection methods} that detect communities by first mapping the graph to a point on the hypersphere, and then projecting this point to the set of clusterings.
In Section~\ref{sec:query_mappings}, we introduced several projection methods that do not correspond to modularity-based methods.
Finally, Section~\ref{sec:experiments} applied our interpretation to real-world networks and demonstrated how our novel projection methods outperforms existing modularity-based methods on several networks.

This work opens up many avenues for future research. In the remainder of this section, we discuss the ones that are most promising in our opinion.

\paragraph{Finding a suitable query vector.}
There are infinitely many possible query mappings and finding the most suitable one is a daunting task. The described geometry allows to split the selection of a query vector into two subtasks: first finding a suitable meridian and then finding a suitable latitude on that meridian. For the first problem, a reasonable approach could be to search for the meridian that minimizes the correlation distance to the ground truth clustering. Since the correlation distance is the angle between the meridians (recall Theorem~\ref{thm:meridian_correlation}), it does not depend on the query latitude that we choose in the second step. However, we have shown in Section~\ref{sec:realworld} that this criterion does not always result in better performances. The problem of finding the optimal latitude on a given meridian seems simpler, as it is one-dimensional.

\paragraph{Projections in other geometries.} We proved that maximizing modularity is equivalent to minimizing the angular distance to a query vector. This begs the question whether minimizing other distances would also give good community detection methods. In particular, it would be interesting to investigate minimizing the correlation distance instead of the angular distance, as this distance would allow one to use the triangle inequality to bound the correlation distance between candidate and ground truth clusterings by their correlation distances to the query vector. Furthermore, since the correlation distance is a distance between meridians (as proven by Theorem~\ref{thm:meridian_correlation}), such methods would be invariant to the choice of the query latitude.

\paragraph{Improved modularity optimization algorithms.}
Finally, as we have proven that the Louvain algorithm is essentially an approximate nearest-neighbor algorithm, it would be interesting to see whether existing approximate nearest-neighbor algorithms outperform the Louvain algorithm in terms of the obtained modularity value or in terms of running time. At any rate, it may be worth-wile to investigate whether the geometry may be utilized to improve upon Louvain or other modularity-optimizing algorithms.

\paragraph{Acknowledgements.}
This work is supported in part by the Netherlands Organisation for Scientific Research (NWO) through the Gravitation {\sc NETWORKS} grant no.\ 024.002.003.

\vskip 0.2in
\bibliography{main}

\appendix

\section{Additional proofs}\label{apx:proofs}
In this appendix, we derive the latitude of the ER modularity vector and prove that the Louvain algorithm is a projection.

\begin{lemma}
    The latitude of the modularity vector for the Erd\H{o}s-R\'{e}nyi null model is given by
    \[
    \tan\ell(\vec{q}_M^{\textrm{ER}}(G;\gamma))=\frac{\sqrt{\frac{N-m_G}{m_G}}}{\gamma-1}.
    \]
\end{lemma}
\begin{proof}
    Note that we have $\vec{p}^{\textrm{ER}}(G)=\tfrac{m_G}{N}\vec{1}$ so that $\|\vec{p}^{\textrm{ER}}(G)\|^2=\tfrac{m_G^2}{N}$ and
    \[
        \langle\vec{p}^{\textrm{ER}}(G),\vec{e}(G)\rangle=2\frac{m_G^2}{N}-m_G.
    \]
    We substitute these into \eqref{eq:latitude_modularity_vector} and rewrite the result to
    \begin{align*}
    \cos\ell(\vec{q}_M^{\textrm{ER}}(G;\gamma))&=\frac{(\gamma-1)m_G}{\sqrt{N}\cdot\sqrt{(1-\gamma)m_G+\gamma^2\frac{m_G^2}{N}-\gamma(2\frac{m_G^2}{N}-m_G)}}\\
    &=\frac{(\gamma-1)\frac{m_G}{N}}{\sqrt{\left((\gamma-1)\frac{m_G}{N}\right)^2+\frac{m_G(N-m_G)}{N^2}}}.
    \end{align*}
    Therefore, the tangent is given by
    \[
    \tan\ell(\vec{q}_M^{\textrm{ER}}(G;\gamma))=\frac{\sqrt{1-\cos^2\ell(\vec{q}_M^{\textrm{ER}}(G;\gamma))}}{\cos\ell(\vec{q}_M^{\textrm{ER}}(G;\gamma))}=\frac{\sqrt{\frac{m_G(N-m_G)}{N^2}}}{(\gamma-1)\frac{m_G}{N}}=\frac{\sqrt{\frac{N-m_G}{m_G}}}{\gamma-1},
    \]
    as required.
\end{proof}

\begin{lemma}\label{lem:louvain-projection}
    The Louvain algorithm is a projection. That is, $\mathcal{L}(\mathcal{L}(\vec{q}))=\mathcal{L}(\vec{q})$ holds for any query vector $\vec{q}\in\mathbb{R}^N$.
\end{lemma}
\begin{proof}
We equivalently prove that for any clustering $C$, Louvain maps the query vector $\vec{q}=\vec{b}(C)$ to itself, i.e., $\mathcal{L}(\vec{b}(C))=\vec{b}(C)$.
Recall that the Louvain algorithm is initialized at the fine pole, i.e., the clustering consisting of $n$ singleton clusters. Then, it iterates through all vertices and relabels each vertex greedily. That is, a vertex is assigned to the cluster that results in the largest decrease in the angular distance to the query vector, or equivalently, the largest increase of $\langle\vec{b}(C'),\vec{b}(C)\rangle$, where $C'$ is the new candidate clustering after relabeling.

Let us consider one cluster $c\in C$. The first time that the iteration of Louvain encounters a vertex $i\in c$, all of its cluster-members are assigned to singleton clusters, so that $i$ is relabeled to the cluster of any of its cluster-members $j\in c\setminus\{i\}$ arbitrarily.
Then, the next time a vertex $k\in c\setminus\{i,j\}$ is encountered in the iteration, the greedy choice is to relabeled $k$ the cluster $\{i,j\}$, as it results in a larger increase than relabeling $k$ to any of the singleton clusters of $c\setminus\{i,j,k\}$. Similarly, all other vertices of $c$ are relabeled to this cluster so that the resulting Louvain candidate contains the cluster $c$. In the same way, all other clusters $c'\in C$ are obtained so that indeed $\mathcal{L}(\vec{b}(C))=\vec{b}(C)$.
\end{proof}

\section{Heatmap of dolphins network}\label{apx:dolphins}
Figure~\ref{fig:dolphins_heatmap} shows a similar experiment as in Section~\ref{sec:beyond_null} performed on Lusseau's network of bottlenose dolphins~\citep{lusseau2004identifying}. Similar to the karate network, this experiment shows that there is a small region of query vectors for which perfect recovery is achieved. 
This time, this is found on meridians that correspond to CM modularity for a negative resolution parameter: one of these query vectors is given by $\vec{q}=\mathcal{P}_{0.58\pi}(\vec{q}_M^{\textrm{CM}}(G;-0.2))$, i.e., CM modularity with resolution $\gamma=-0.2$ projected to the latitude $0.58\pi$.
Note that if we were to try and interpret this query vector in terms of a null model, the expected number of edges between two vertices $i,j$ would be of the form $c_1-c_2d_i^{(G)}d_j{(G)}$ for $c_1,c_2>0$. For some vertex-pairs, this value may be negative, so that it cannot be interpreted as an expected number of edges. The reason that this query vector performs well on this network seems to be that all high-degree vertices are part of the same community. Therefore, vertex-pairs for which $d_i^{(G)}d_j^{(G)}$ is high are more likely to be community members.

\begin{figure}
    \centering
    \includegraphics[width=0.7\textwidth]{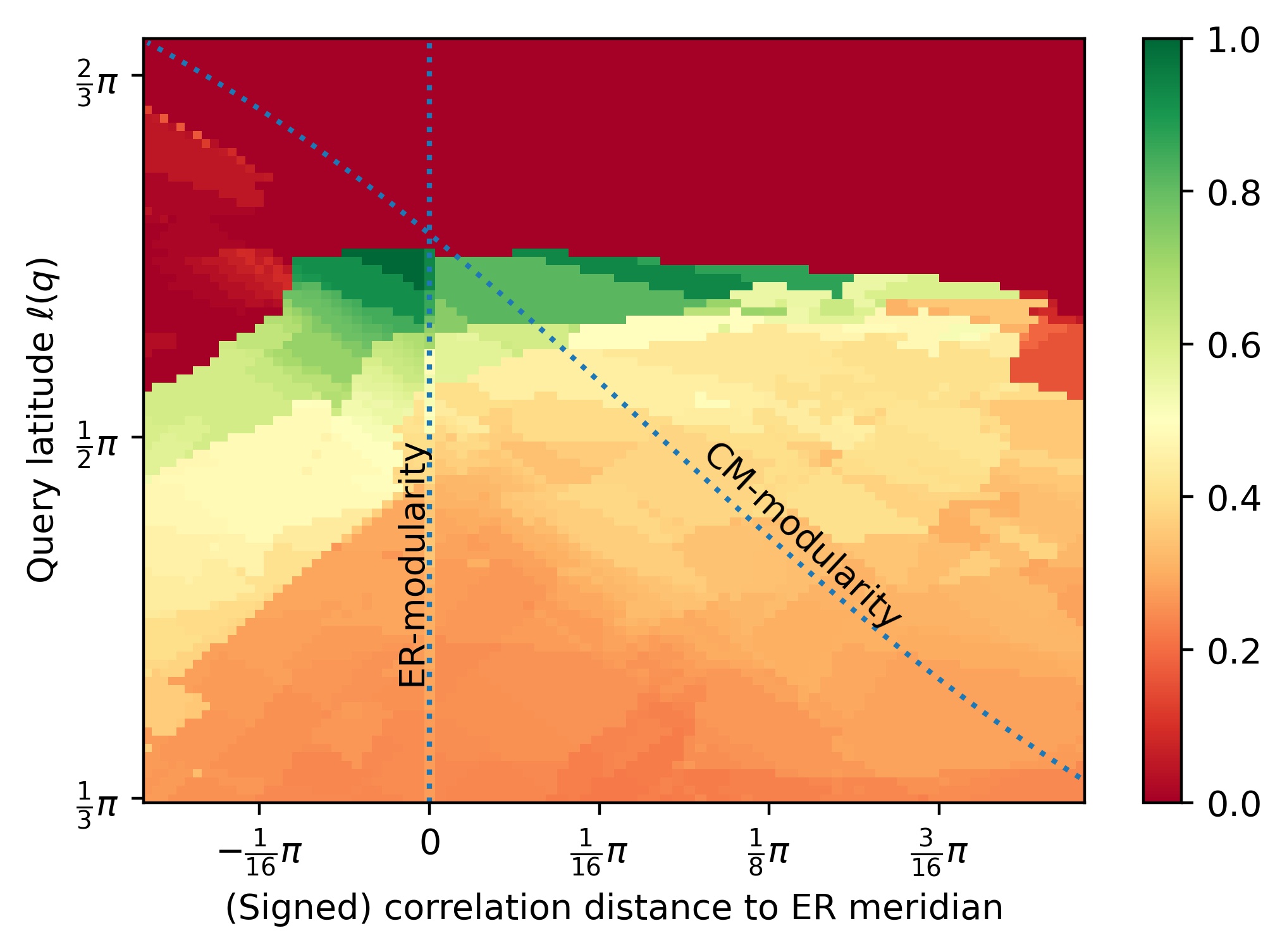}
    \caption{Heatmap of the correlation between ground truth and candidate clusterings for the Dolphins network~\citep{lusseau2004identifying}. We take query vectors from the meridians that $\left(\vec{q}_M^{\textrm{CM}}(G;\gamma)\right)_{\gamma\in[-1.5,2]}$ runs through and vary the query latitude between $\tfrac{1}{3}\pi$ and $\tfrac{2}{3}\pi$. The horizontal coordinates are given by $\textrm{sgn}(\gamma)\cdot\dCC(\vec{q},\vec{e}(G))$, i.e. the (signed) correlation distance to the ER meridian.}
    \label{fig:dolphins_heatmap}
\end{figure}

\end{document}